\documentclass{article}

\pdfoutput=1

\usepackage[margin=.8in]{geometry}

\usepackage{authblk}

\usepackage{pifont}
\usepackage{cite}
\usepackage{amsmath,amssymb,amsfonts}
\usepackage{graphicx}
\usepackage{textcomp}
\usepackage{xcolor}
\def\BibTeX{{\rm B\kern-.05em{\sc i\kern-.025em b}\kern-.08em
		T\kern-.1667em\lower.7ex\hbox{E}\kern-.125emX}}

\usepackage{algpseudocode}
\usepackage{bm}
\usepackage{comment}

\usepackage{amsthm}
\usepackage{algorithm}
\usepackage{blindtext}
\usepackage{multirow}
\usepackage{yfonts}
\usepackage{url}
\usepackage{caption}
\usepackage{subcaption}
\usepackage{balance}

\usepackage{tabularx}
\usepackage{makecell}
\usepackage{hyperref}
\usepackage{float}
\floatstyle{plaintop}
\restylefloat{table}

\ifdefined\VLDB
\newcommand{\subparagraph}{}
\fi

\usepackage{tikz}
\usetikzlibrary{shapes.geometric}
\usepackage{booktabs}
\usepackage{nopageno}

\allowdisplaybreaks

\newtheorem{definition}{Definition}
\newtheorem{lemma}{Lemma}
\newtheorem{theorem}{Theorem}

\newcommand{\lss}{LSS\xspace}

\newcommand{\lsslf}{LSS-LF\xspace}
\newcommand{\lsslfFULL}{Learned Space Saving without Low Frequencies}

\newcommand{\lsslfs}{LSS-LFS\xspace}
\newcommand{\lsslfsFULL}{Learned Space Saving Low Frequency Singles}

\newcommand{\lssplus}{LSS+\xspace}

\newcommand{\lsscbf}{LSS-CBF\xspace}

\newcommand{\lssasn}{LSS-HH}
\newcommand{\lssaeFULL}{Learned Space Saving with Heavy Hitters}

\newif\ifcomm
\commtrue
\ifcomm
\else
\commfalse
\fi
\ifcomm
	\newcommand{\mycomm}[3]{{\footnotesize{{\color{#2} \textbf{[#1: #3]}}}}}
	\newcommand{\CRdel}[1]{\textcolor{red}{\sout{#1}}}
\else
    \newcommand{\mycomm}[3]{}
    \newcommand{\CRdel}[1]{}
\fi

\newcommand{\BoxedState}[1]{
    \State \fbox{\parbox[t]{\dimexpr\linewidth-\algorithmicindent-2em}{#1}}%
}

\ifdefined\EXTENDED
\newcommand{\matrixCellWidth}{5.8cm}
\else
\newcommand{\matrixCellWidth}{5.5cm}
\fi
\newcommand{\smatrixCellWidth}{3cm}

\newcommand{\mmatrixCellWidth}{3.2cm}

\newcommand{\ssmatrixCellWidth}{2.5cm}

\newcommand{\legendCellWidth}{8cm}

\newcommand{\eps}{\epsilon}

\newcommand{\epsNError}[1][N]{\ensuremath{#1 \epsilon}}

\newcommand{\oneOverE}{ \ensuremath{\eps^{-1}} }

\renewcommand{\gamma}{\phi}

\begin{document}

\date{}

%\title{Learned Space Saving}
%\title{Learned Space Saving for Frequency and Frequent Items Estimation}

\title{Learning-Based Heavy Hitters and Flow Frequency Estimation in Streams}

\author[1]{Rana Shahout}

\author[1]{Michael Mitzenmacher}

\affil[1]{Harvard University, USA}

%\title{Learning-Based Heavy Hitter Identification in Streams}

%\title{LSS: Learned Frequency and Frequent Items Estimation}

%%
%% The "author" command and its associated commands are used to define the authors and their affiliations.


\maketitle

\begin{abstract}
Identifying heavy hitters and estimating the frequencies of flows are fundamental tasks in various network domains. Existing approaches to this challenge can broadly be categorized into two groups, hashing-based and competing-counter-based.
The Count-Min sketch is a standard example of a hashing-based algorithm, and the Space Saving algorithm is an example of a competing-counter algorithm.  
Recent works have explored the use of machine learning to enhance algorithms for frequency estimation problems, under the algorithms with prediction framework. However, these works have focused solely on the hashing-based approach, which may not be best for identifying heavy hitters.

In this paper, we present the first learned competing-counter-based algorithm, called \lss{}, for identifying heavy hitters, top $k$, and flow frequency estimation that utilizes the well-known Space Saving algorithm. We provide theoretical insights into how and to what extent our approach can improve upon Space Saving, backed by experimental results on both synthetic and real-world datasets. Our evaluation demonstrates that \lss{} can enhance the accuracy and efficiency of Space Saving in identifying heavy hitters, top $k$, and estimating flow frequencies.
\end{abstract}

\section{Introduction}
\label{sec:intro}

\begin{figure}[t]
  \centering
  \includegraphics[width=0.5\linewidth]{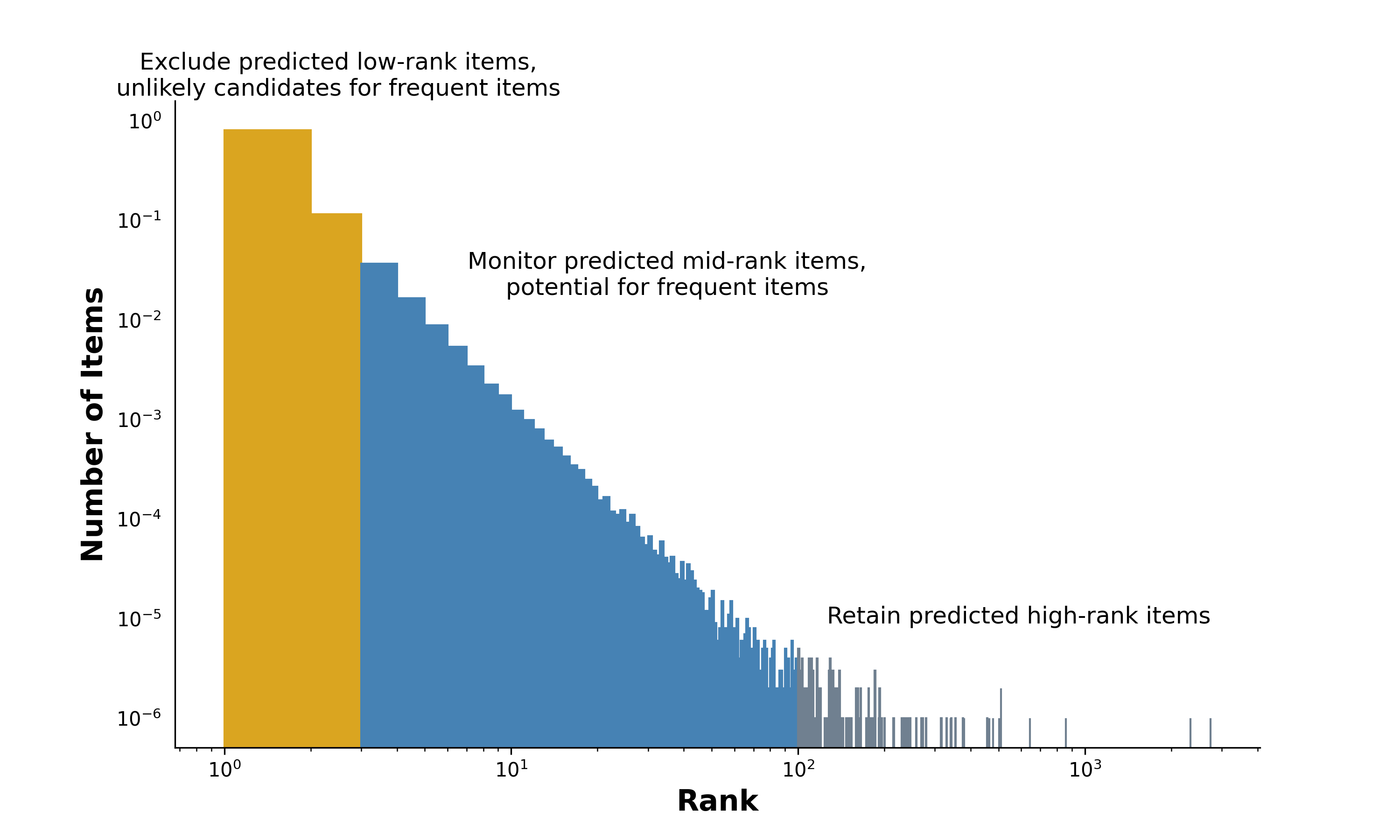}
  \caption{In many practical distributions (such as Zipfian, shown here using logarithmic scales on both axes), there are many low-frequency items.}
  \label{fig:zipf_distribution}
\end{figure}

The paradigm of learning-augmented algorithms, also known as algorithms with predictions, combines machine learning and traditional algorithms. This approach aims to enhance algorithms and data structures by leveraging predictions from machine learning models. Such learning-augmented algorithms have demonstrated theoretical and empirical benefits across numerous areas, including scheduling~\cite{DBLP:conf/soda/AzarLT22,DBLP:conf/stoc/AzarLT21,DBLP:conf/acda/Mitzenmacher21,DBLP:conf/innovations/ScullyGM22}, caching~\cite{lykouris2021competitive,DBLP:conf/soda/Rohatgi20}, and approximate frequency estimation~\cite{hsu2019learning}. In the realm of approximate frequency estimation, the goal is to utilize predictive models to improve the accuracy of the estimation.

%Estimating how often certain items appear in a data stream and identifying the most frequently occurring items, also known as heavy hitters, are fundamental tasks in data analysis. They are used in many areas like database management~\cite{cormode2004holistic, fang1998computing, ting2018data}, load balancing~\cite{dittmann2002network}, network monitoring~\cite{basat2020designing, harrison2020carpe}, intrusion detection~\cite{mukherjee1994network} and caching~\cite{einziger2017tinylfu, zhang2021cocosketch}.
A stream of packets flowing through a network or a specific switch can be divided into multiple flows, each identified by a unique set of attributes. For example, the 5-tuple consists of source IP, source port, destination IP, destination port, and protocol is often used as a flow identifier. Two fundamental problems in network statistics are determining the frequency of a flow, i.e., counting the number of packets with a given flow identifier that passes through the network or switch, and identifying the heavy hitter flows (the most frequent ones), as these flows often have the most significant impact. These monitoring capabilities are desirable for various network algorithms and domains, such as load balancing~\cite{dittmann2002network}, routing, fairness~\cite{kabbani2010af}, intrusion detection~\cite{mukherjee1994network}, ~\cite{garcia2009anomaly}, caching~\cite{einziger2017tinylfu}, and policy enforcement~\cite{sommers2007accurate}.

In a modern network, as the number of flows can be massive, providing precise answers to queries regarding flow\footnote{The terms flow and item are used interchangeably.} frequencies and identifying heavy hitters\footnote{The terms heavy hitters and frequent items are used interchangeably.} is often prohibitively costly. To that end, streaming algorithms process the data in a single pass while efficiently estimating flow frequencies with minimal storage space are essential. These algorithms build a {\em sketch}; a compact data structure that extracts statistical information in one pass on the entire data stream.

%{\bf Let's talk about this paragraph;  I think Space Savings is also a sketch, so I'd refer to these as a "hashing-based counters" approach and a "competing-items counters" approach.  Unless these names are used elsewhere in the literature.  Note also the "competing-items" approach is not exactly a frequency estimation algorithm;  it just focuses on heavy hitters, though you could return 0 or some other fixed estimate for other items.  Anyhow I still think there's better terminology/description here.}  
%{\bf MM:  we use "sketch" a lot -- if we want to explain how we mean the term, do so somewhere early?}
Standard frequency estimation algorithms can be broadly classified into two categories: hashing-based and competing-counter-based approaches. The hashing-based approach~\cite{cormode2005improved, charikar2002finding, jin2003dynamically} hashes data items into buckets, potentially leading to collisions among different items. 
It then calculates item frequencies based on the number of items hashed into each bucket.
In contrast, the competing-counter-based approach~\cite{metwally2005efficient, misra1982finding} maintains a fixed number of counters that at any time tracks a subset of the input items, without hashing and corresponding collisions. These data stream items ``compete'' for the limited counters, with the algorithm aiming to allocate counters to the most frequent items.
For identifying heavy hitters and top $k$, while hashing-based approaches require additional structures like heaps, competing-counter-based approaches do not necessitate such auxiliary components, making them more space-efficient.
Indeed, competing-counter-based algorithms have been shown to be more space-efficient than hashing-based algorithms, both empirically and asymptotically~\cite{cormode2008finding, cormode2010methods}, as hashing-based algorithms require allocating significantly more counters and space than theoretical lower bounds to mitigate the impact of hash collisions on counter accuracy. Additionally, competing-counter-based algorithms offer deterministic error guarantees, in contrast to hashing-based, which are randomized.
%Additionally, competing-counter-based algorithms offer deterministic error guarantees, contrasting with the randomized nature of hashing-based algorithms.
%{\bf MM: this part above is unclear-- hashing based algorithms benefit from static memory allocation -- but counter-based algorithms don't?  Why not?  I would just make more general statements ---}
%The choice between these two approaches often depends on the specific requirements and constraints of the application at hand.

Hsu et al.~\cite{hsu2019learning} introduced a learned hashing-based algorithm for frequency estimation that utilizes a heavy hitter predictor. When an incoming item is predicted to be a heavy hitter, it is assigned a unique bucket for accurate counting. Items not predicted as heavy hitters are processed using a conventional sketching algorithm. This approach reduces collisions between heavy hitters and less frequent items, leading to improved overall accuracy. It is important to note that heavy hitter predictors may exhibit errors, and thus relying solely on them to identify heavy hitters can violate error guarantees.
To our knowledge, there is a lack of research on learned competing-counter-based algorithms. Unfortunately, directly applying the same strategy to competing-counter-based algorithms by simply assigning all counters to predicted heavy hitters leads to unbounded error estimation. For instance, a falsely predicted heavy hitter item occupies a counter that could have been used for a real heavy hitter.  (Recall that
competing-counter-based approaches have a fixed number of counters.)
Even worse, a true heavy hitter might be mistakenly identified as a non-heavy hitter and consequently ignored. Unlike learned hashing-based algorithms that focus only on heavy-hitter items, we propose an approach that employs a predictor to filter out noise caused by items predicted to have low frequencies as in many network traces~\cite{CAIDACH16} and practical distributions (e.g. Zipfian), many items are with one occurrence. (Figure~\ref{fig:zipf_distribution}).
%our approach for the learned competing-counter-based algorithms employs a predictor that targets not only those heavy hitters but also the other extreme - low-frequency items.
%Unlike in the hashing-based approach, where a heavy hitter could still be accounted for through shared counters in a sketch, there is no fallback to the standard algorithm in the competing-counter-based case.
%{\bf MM:  This seems a bit much to claim -- aren't we sort of doing the same thing here -- if you're predicted heavy, you can get your own counter.  If not, we are going to put you in the competing-counter structure as a fallback.}

In this paper, we introduce a learned competing-counter-based algorithm
for identifying heavy hitters, top $k$, and flow frequency estimation. We focus on the Space Saving algorithm, although our approach could be applied to any competing-counter-based algorithm, such as the MG algorithm~\cite{misra1982finding}. We present Learned Space Saving (LSS), a novel technique that leverages machine learning predictions to guide the competition for limited counters among data items. 
Specifically, LSS aims to exclude predicted ``weak'' or low-frequency items from the competition, while ensuring that predicted ``strong'' or heavy-hitter items remain in the competition. In this way, the Space Saving accuracy is improved.

%{\bf MM: I don't think we've clearly stated at this point that we want to make use of 2 types of predictors, that's a main point we should emphasize strongly when introducing LSS, distinguishes from prior work.}
LSS is designed to be resilient against prediction errors. Our LSS method employs a filtering mechanism to exclude predicted low-frequency items, resulting in a ``cleaner'' data structure.
To ensure robustness against incorrectly predicting a high-frequency item as a low-frequency item, we employ a Counting Bloom filter that tracks these items and ensures an item is not consistently ignored just because it is (repeatedly) predicted as low frequency.  
% encodes previously overlooked items. 
When using the heavy hitter predictor, we divide the counters into fixed and mutable counters. A predicted heavy hitter can be allocated a fixed counter, to achieve an accurate count. But fixed counters are limited, so if there are excessive incorrect heavy hitter predictors, the mutable counters allow for a standard Space Saving algorithm on items after the fixed counters are filled.
We break LSS into two variants, \lsslf{} and \lssasn{}, each utilizes a distinct learning model and is designed to exhibit resilience against prediction errors.

\begin{comment}
Our contributions are as follows:
\begin{itemize}
    \item 
We propose a learning-based approach, LSS, for frequency and frequent item estimation that improves the accuracy of the Space Saving algorithm. We design LSS to be robust against prediction errors (Section~\ref{sec:lss}).
    \item 
We break down LSS into two components, \lsslf{} (Section~\ref{sec:lss_lf}) and \lssasn{} (Section~\ref{sec:lss_hh}), and show the robustness of each component.
    \item 
We present theorems that provide insight into the potential gains from predictors, considering a scenario in which the predictions are perfect.
    \item 
We propose \lssplus{}, a variation of \lss{} that offers higher update throughput at the cost of relaxing the deterministic approximation guarantees to probabilistic approximation guarantees (Section~\ref{sec:lss_plus}).
    \item 
We implement and evaluate \lss{} and its variants (Section~\ref{sec:eval}) using a synthetic dataset and two real-world datasets: traffic load on an Internet backbone link and web search. We observed accuracy improvements under specific configurations. These enhancements varied based on the dataset, parameters, and used memory. Notably, LSS demonstrated up to $18\%$ better precision in identifying the top-k ($k=64$) most frequent items, $24\%$ higher recall for heavy hitters, and a $34\%$ reduction in RMSE for frequency estimation, compared to SS.
\end{itemize}
\end{comment}

Our contributions are: 
\begin{itemize}
    \item We propose \lss{}, a learning-based approach for identifying heavy hitters, top $k$ and frequency estimation, improving Space Saving's accuracy and designed to be robust against prediction errors (Section~\ref{sec:lss}).  \item We break down \lss{} into \lsslf{} (Section~\ref{sec:lss_lf}) and \lssasn{} (Section~\ref{sec:lss_hh}), showing each component's robustness.  \item We present theorems providing insight into potential gains  \item We propose \lssplus{}, a variation of \lss{} offering higher update throughput by relaxing deterministic to probabilistic approximation guarantees (Section~\ref{sec:lss_plus}).  \item  We implement and evaluate \lss{} and variants (Section~\ref{sec:eval}) on synthetic and real-world datasets (Internet traffic, web search); \lss{} demonstrated up to $18\%$ better top-k precision, $24\%$ higher heavy hitter recall, and $34\%$ lower RMSE for frequency estimation compared to Space Saving under certain configurations.

\end{itemize}

\section{Background}
\label{sec:background}

\subsection{Preliminaries}
Given a \emph{universe} $\mathcal{U}$, a \emph{stream} $\mathcal{S} = u_1, u_2, \ldots \in \mathcal{U}^N$ is a sequence of arrivals from the universe. (We assume the stream is finite here.)
We define the frequency of an item $i$ in $\mathcal{S}$ as the number of
items corresponding to $i$ in $\mathcal{S}$ and denote this quantity by $f_i$.

%{\bf MM:  Say something about weighted operations -- stream could be ordered pairs $(i,c)$ where $c$ is added to the count, we're assuming additive model.}

We seek algorithms that support the following operations:
\begin{itemize}
	\item {\sc \textbf{ADD}$\bm{(i)}$}: given an element $i\in\mathcal U$, append $i$ to $\mathcal S$.
	\item {\sc \textbf{Query}$\bm{(i)}$}: return an estimate $\widehat{f_i}$ of $f_i$
\end{itemize}

A weighted stream consists of tuples of the form $(u_i, w_i)$, where $u_i$ represents the item's id and $w_i$ its (non-negative) weight. At each step, a new tuple is added to the stream. In this setting, the weight $w_i$ is added to the corresponding item's frequency. Frequency estimation algorithms typically assume unweighted updates, such as click streams, while others assume weighted updates, such as network traffic volumes. In this paper, we focus on the unweighted updates model for ease of exposition, although our approach applies to weighted data streams as well, with straightforward modifications.
%{\bf MM: I think we need some statement here of our work applies to weighted streams but we focus on unweighted streams throughout for ease of expositions;  modifications are straightforward?}

%We use the following standard formalization for approximate frequency estimation problems.  

%{\bf MM: modify definition to add randomized versions have error guarantees.}

%\rana{should we have the probability in the problem definition? I mean $\epsilon$-Frequency for competing-counter-based and $\epsilon, \delta$-Frequency for hashing-based? if so it's easier to have two definitions..}

%{\bf MM: check if new statements match what we want.}
\begin{definition}
\label{definition:frequency_error_guarantee}
An algorithm solves $\epsilon$-Frequency if given any Query$(i)$ it returns $\widehat{f_i}$ satisfying
$$f_i \le \widehat{f_i} \le f_i  + \epsNError.$$ 
A (randomized) algorithm is said to solve $\epsilon$-Frequency with probability $1-\delta$ if, for any $i$
chosen before the stream is processed, 
Query$(i)$ returns $\widehat{f_i}$ satisfying the above bound with probability $1-\delta$. Similarly, a (randomized) algorithm solves $\epsilon$-Frequency in expectation if given any Query$(i)$ it returns $\widehat{f_i}$ satisfying 
$$f_i \le E[\widehat{f_i}] \le f_i  + \epsNError.$$
\end{definition}

\begin{definition}
\label{definition:hh_error_guarantee}
An algorithm solves $(\epsilon, \theta)$-HeavyHitters if it returns a set of items $B$, such that for every item $i$: if $f_i \ge \theta N$, then $i \in B$, and if $f_i \le (\theta - \epsilon) N$ then $i \notin B$. 
An algorithm solves $(\epsilon, \theta)$-HeavyHitters with probability $1-\delta$ if it returns a set of items $B$ satisfying the above conditions with probability $1-\delta$.
\end{definition}

Deterministic solutions (competing-counter-based)~\cite{metwally2005efficient, misra1982finding} for the $\epsilon$-HeavyHitters problem ensure the identification of all items with sufficiently large counts, but potentially may include some items with counts smaller than the given heavy hitter threshold. 
In contrast, hashing-based~\cite{cormode2005improved, charikar2002finding} for the $\epsilon$-HeavyHitters problem may introduce a probability of failure.% ($\delta$).  

\begin{comment}
    
{\bf MM: maybe the subsection below goes before/into the definitions section above?}

%\subsection{Deterministic and Randomized Solution Guarantees}
\subsection{Competing-Counter-Based and Hashing-Based Guarantees}
{\bf MM: What is the exact set of frequent elements -- do you mean one of the above problems?}
When dealing with limited memory and large data streams, determining the exact set of frequent elements requires  $\Omega(N)$ space~\cite{cormode2008finding} and thus becomes infeasible.
An alternative and more practical approach to big data is to use approximation techniques. Deterministic solutions (competing-counter-based)~\cite{metwally2005efficient, misra1982finding} for the $(\epsilon, \theta)$-FrequentItems problem ensure the identification of all items with sufficiently large counts, but potentially may include some items with counts smaller than the given heavy hitter threshold. 
In contrast, hashing-based~\cite{cormode2005improved, charikar2002finding} for the $(\epsilon, \theta)$-FrequentItems problem may introduce a probability of failure.  
%{\bf MM: below seems redundant, people should know.  Let's cut?}
%Deterministic algorithms provide robust guarantees, ensuring the identification of all significant elements. In contrast, randomized algorithms take a best-effort approach, aiming to report all heavy items while avoiding lightweight items, but without absolute certainty.
\end{comment}

%\ifdefined\EXTENDED

\begin{table}[t]

	\centering
	\small
	\caption{List of Symbols}
	\begin{tabular}{|c|p{6.6cm}|}
		\hline
		Symbol & Meaning \tabularnewline
		\hline
		$\mathcal S$ & The data stream \tabularnewline
		\hline
		$\mathcal U$ & The universe of elements \tabularnewline
		\hline
		$N$ & The stream size \tabularnewline
		\hline
            $u_{[x...y]}$ & The partial stream $u_x, u_{x+1},\cdots,u_y$ \tabularnewline
            \hline
		%$f_i^{W}$ & Item $i$ frequency over the last $W$ arrivals in $\mathcal S$ \tabularnewline
		%\hline
            %$f(u_{x...y})$ & frequency vector after processing $u_{x...y}$ \rana{delete if we move theorem4 to appendix}\tabularnewline
            %\hline
            %$c_i(u_{1\cdots x})$ & SS counter value of item $i$ after processing $u_{1...x}$ \rana{delete if we move theorem4 to appendix}\tabularnewline
            %\hline
		$\epsilon$ & An estimate accuracy parameter \tabularnewline
		\hline
		$\theta$ & The heavy hitters threshold \tabularnewline
		\hline
		$\tau$ & \lssplus{} robustness probability \tabularnewline
		\hline
		$t$ & \lsslf{} threshold \tabularnewline
		\hline
            $\delta(u_{1\cdots x})$ & SS error vector after processing $u_{1...x}$ \tabularnewline
            \hline
	    $\ell$ & Number of single occurrence items\tabularnewline
            \hline
            $fpr$ & Bloom filter false positive rate \tabularnewline
            \hline
	\end{tabular}
	\label{tbl:notations}
\end{table}
%\fi

\subsection{Robustness in Learning-Augmented Algorithms}
As mentioned, a learned-augmented algorithm combines traditional algorithms with machine learning models. 
(It should be noted, however, that this approach generally treats the machine learning models as black boxes, allowing it to work with any model that provides usable predictions.)
However, machine learning methods are inherently imperfect and may exhibit errors, including substantial and unexpected errors.
A key question is how can we use predictions while maintaining robustness, which refers to ability of an algorithm to maintain reasonable performance even if the predictions are simply wrong \cite{lykouris2021competitive,mitzenmacher2022algorithms}. Ensuring robustness is essential because machine learning models are rarely perfect in practice. There are several reasons why learned models may exhibit errors. First, most models are trained to perform well on average by minimizing expected losses. In doing so, they may reduce errors on the majority of inputs at the expense of increasing errors on outlier cases. Additionally, generalization error guarantees only hold when the training and test samples are drawn from the same distribution. If this assumption is violated, due to distribution drift or adversarial samples, the predictions can vary from the ground truth. 

%{\bf MM: check this.}
One general approach for learned-augmented algorithms to try to achieve robustness is to fall back on the traditional algorithm when the model is inaccurate.  This requires being able to notice inaccurate predictions and change to the traditional algorithm quickly and effectively. Another related approach, which we use here, is to find ways to limit the damage that can be caused by incorrect predictions by using additional algorithm or data structure.   

\subsection{Space Saving}
The Space Saving (SS) algorithm~\cite{metwally2005efficient} is a competing-counter algorithm that provides frequency estimation for data stream items.
%We first review Space Saving. 
Space Saving maintains a set of $k$ entries, denoted by $T$, each entry has an associated item $i$ and counter, and we use $c_i$ to denote the counter value associated with item $i$ if any exists.
When $k=\frac{1}{\epsilon}$, Space Saving estimates the frequency of any item with an additive error less than $N\epsilon$ where $N$ is the stream size.% (as stated below formally in Lemma~\ref{lemma4}).

%{\bf MM: below assumes unweighted streams, why I suggested stating that is our default earlier.}
When an item $i$ is encountered within the stream that is in the set $T$, the algorithm increments its corresponding counter $c_i$.
In cases where the item $i$ is not present in $T$ and the size of $T$ is less than $k$,
the algorithm adds $i$ to the set $T$ and initializes its count to $1$ ($c_i=1$).
Otherwise, when the item $i$ is not in $T$ and the size of $T$ has reached $k$, Space Saving identifies an item $j$ within $T$ with the minimum non-zero count $c_j$, denoted by $minCount$.
The algorithm then executes an update in which $c_i$ is assigned the value of $c_j + 1$, and the
set $T$ is updated to replace item $j$ with item $i$.
% This effectively results in the replacement of item $j$ with item $i$ within the set $T$.
To estimate the frequency of an item, if the item is in $T$ then we report its count.
Otherwise, we report the smallest counter stored in $T$.
% as an upper bound on the true value, or $0$ as a lower bound on the true value.
%{\bf MM: Do we need to state the 0 lower bound?}
Pseudocode for SS is presented (in black) in Algorithm~\ref{alg:lss}.

Space Saving satisfies the following properties (the first two
properties are proved in~\cite{metwally2005efficient} while the latter is proved in~\cite{zhao2021spacesaving}):

\begin{lemma}
\label{lemma2}
    Space Saving with $k=\epsilon^{-1}$ counters ensures that after processing $N$ insertions, the minimum count of all monitored items is no more than $\frac{N}{k} = N\epsilon$, i.e, $minCount < N\epsilon$.
\end{lemma}

\begin{lemma}
    All items with frequency larger than or equal to $minCount$ are in the set $T$, the set of items with associated counters.
\end{lemma}

\begin{lemma}
\label{lemma4}
    Space Saving with $\epsilon^{-1}$ counters can estimate the frequency of any item with an additive error less than $N\epsilon$.
\end{lemma}

\section{Learned Space Saving}
\label{sec:lss}

\begin{figure}[t]
    \centering
    \includegraphics[width=0.5\linewidth]{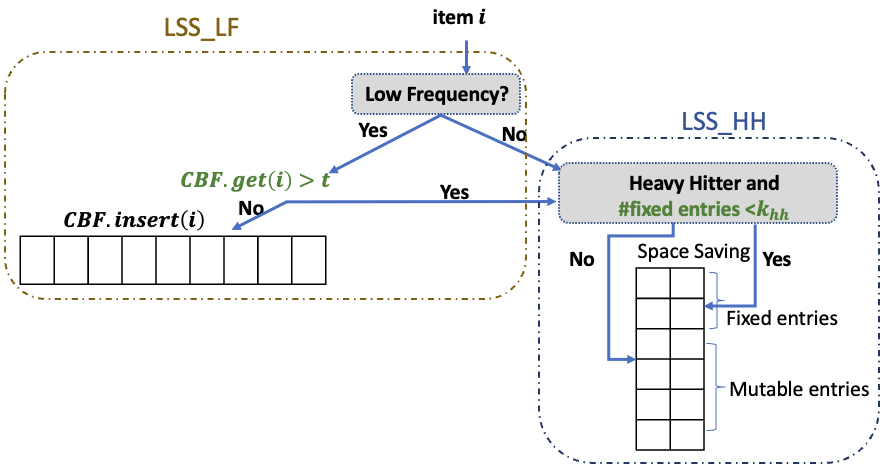}
    \caption{Overview of \lss{}, combining the approaches of \lsslf{} and \lssasn{}. \lsslf{} utilizes a low-frequency predictor to exclude infrequent items. \lssasn{} divides entries into fixed and mutable entries, using a heavy hitters predictor to allocate fixed entries for frequent items while processing remaining items through the mutable entries. To mitigate the impact of prediction errors, \lsslf{} utilizes a Counting Bloom Filter (CBF), while \lssasn{} sets a limit ($k_{hh}$) on the number of fixed entries (represented in green color).}
    \label{fig:LSS_Overview}
\end{figure}

\begin{algorithm}[t]
\caption{\lss{}} \label{alg:lss}
\begin{algorithmic}[1]
\State \textbf{Initialization}: $T \leftarrow \varnothing$,
\State \textcolor{brown}{$CBF \textit{-- Counting Bloom filter}$, $t \textit{-- CBF threshold}$}
\State \textcolor{brown}{$\textit{freqPredictor}$ -- frequency predictor}
\State \textcolor{blue}{$fixedEntries \leftarrow 0$}
\State \textcolor{blue}{HH - heavy hitters predictor, $k_{hh}$ - number of allowed fixed entries}

    \Function{Insert}{$i$}
    \textcolor{brown}{\If{\textcolor{brown}{$\textit{freqPredictor(i)} > t$ or} $CBF.GET(i) > t$} \label{lsscbf:line_insertSS}
        \textcolor{black}{
        \If{$i \in T$}
            \State $c_i \leftarrow c_i + 1$
        \ElsIf{$|T| < k$}
            \State $T \leftarrow T \cup \{i\}$
            \State $c_i \leftarrow 1$
            \textcolor{blue}{\If{HH(i) and $fixedEntries \leq  k_{hh}$}
            \State mark entry of $c_i$ as fixed entry
            \State $fixedEntries \leftarrow fixedEntries + 1$
            \EndIf
            }
        \Else
            \State $j \leftarrow \text{argmin}_{j \in \textcolor{blue}{T_{unasn}}} c_j$ \label{line:lssasn_minarg}
            \State $c_i \leftarrow c_j + 1$ \label{line:lssasn_replacing}
            \State {$T \leftarrow T \cup \{i\} \setminus \{j\}$}
            \textcolor{blue}{\If{HH(i) and $fixedEntries \leq  k_{hh}$}
            \State mark entry of $c_i$ as fixed entry
            \State $fixedEntries \leftarrow fixedEntries + 1$
            \EndIf}
        \EndIf
        }   
    \Else
    \State $CBF.ADD(i)$ \label{lsscbf:line_insertcbf}
    \EndIf}
    \EndFunction

    \Function{Query}{$i$}
        \If{$i \in T$}
            \State \Return $c_i$ \textcolor{brown}{$+ t$} \label{lsscbf:line_returncnt}
        \Else
            \State \Return $\min_{j \in T} c_j$  \textcolor{brown}{$+ t$}\label{lsscbf:line_return1}
        \EndIf
    \EndFunction
\end{algorithmic}
\end{algorithm}

We propose a learning-based approach to enhance the accuracy of the Space Saving algorithm, called Learned Space Saving (LSS). 
Our approach allows for two types of predictors: one for identifying low-frequency items and another for identifying heavy hitters. Either or both can be used. Alternatively, a single stronger predictor capable of predicting item frequencies could be employed. To accommodate the availability of predictors in different systems, we divide \lss{} into two components: \lss{} with low frequency predictor (\lsslf{}) when only the low-frequency predictor is available, and \lss{} with heavy hitter predictor 
(\lssasn{}) when only the heavy hitters predictor is available. When both predictors are available, or when a single predictor can predict item frequencies, we utilize the full LSS, which combines the strengths of the individual components.

\lsslf{} employs a filtering mechanism to exclude low-frequency items, resulting in a ``cleaner'' data structure. It utilizes a Counting Bloom Filter to ensure robustness against prediction errors. %{\bf MM: Don't think we have explained/defined "robustness" yet?}
\lssasn{} takes a different approach by dividing the Space Saving counters into two categories: fixed counters and mutable counters. It uses a heavy hitters predictor to allocate fixed counters for items identified as heavy hitters. The remaining items are processed using the traditional Space Saving algorithm, using mutable counters for maintaining their counts. This overall approach is illustrated in Figure~\ref{fig:LSS_Overview}.

Our design decouples the robustness logic (highlighted in green in Figure~\ref{fig:LSS_Overview}) from the remaining algorithm logic. Figure~\ref{fig:example} illustrates the algorithms presented in this paper. We have employed a color-coding scheme to aid readers in visualizing and distinguishing between the presented algorithms. We represent four algorithms in the pseudo-code of Algorithm~\ref{alg:lss}, each uniquely represented using specific colors, as follows:
Space Saving is represented in black.
\lsslf{} combines black and \textcolor{brown}{brown} colors.
\lssasn{} combines black and \textcolor{blue}{blue} colors.
Finally \lss{} combines all of them.

\section{\lsslf{}: \lsslfFULL{}}
\label{sec:lss_lf}

%\textbf{Intuition.}
%While Lemma~\ref{lemma4} demonstrates that Space Saving with $\frac{1}{\epsilon}$ counters has an additive error bounded by $N\epsilon$, Berinde et al.~\cite{berinde2010space} have shown that better bounds are possible if one knows about the tail of the frequency distribution.  
%{\bf MM: This transition isn't quite clear -- maybe, 
%In this paper, we apply predictions rather than distribution knowledge to provide better bounds.  We claim....}
%In this paper, 
We claim that by selectively removing low-frequency items and properly adjusting their estimation, we can enhance the accuracy of the Space Saving algorithm.
We first consider the case of items appearing only once in a stream and ``remembering'' their count as 1 as a special case.
Then, we present the general case of removing low-frequency items up to $t$ occurrences.

\begin{comment}
We start by specifically considering items with a frequency of 1, and showing that Space Saving accuracy improves by removing single occurrence items.
The following notation is used throughout the discussion of Theorem~\ref{theorem:remove_singles}, which is for this case. We use $u_{[x...y]}$ as a shorthand to represent the partial stream $u_x, u_{x+1},\cdots,u_y$.
Let $f(u_{x...y})$ represent the frequency vector after processing the partial stream $u_{x...y}$.
We define an approximation error vector, denoted by $\delta$.
The error associated with the $i$-th element is given by $\delta_i$, which is calculated as the absolute difference between the true frequency of item $i$, $f_i$, and its corresponding counter $c_i$, i.e., $\delta_i = |f_i - c_i|$.
Additionally, $\delta(u_{1\cdots x})$ denotes the error vector and $c_i(u_{1\cdots x})$ represents the counter value of item $i$ after processing the partial stream $u_{1...x}$.

\begin{definition}
An item $i$ is a \textit{single occurence} of the stream $u_{1\cdots s}$ if $f_i = 1$.  \label{def:signle}
\end{definition}

\begin{theorem}
\label{theorem:remove_singles}
    Once $T$ is filled with $k$ items, the removal of a single occurrence item from the stream improves the accuracy of the Space Saving. Formally,
    if item $i$ is a single occurrence item of the stream that occurs at position $x$, then $\delta(u_{1\cdots (x-1)}v_{1 \cdots y}) < \delta(u_{1\cdots x}v_{1 \cdots y})$
\end{theorem}
The intuition for Theorem~\ref{theorem:remove_singles} 

\end{comment}

Our intuition is that occurrences of items that appear once ``disturb'' the accuracy of items that are in $T$. Space Saving could set a 
counter of an item that is potentially frequent (i.e. heavy hitter) to zero by encountering a single occurrence item.
Eventually, the replaced frequent item will return to the Space Saving table (from correctness), but in this case, the ``counter history'' will be lost, causing inaccuracy in the counters.

\ifdefined\arXiv
\begin{proof}
We prove this using induction on $t$.
\textbf{Base case at $t=0$:} As $u_x = i$ 
which is a single occurrence item, $c_i(u_{1\cdots (x-1)}) = 0$.
At the time of arrival of $u_x$, the algorithm of Space Saving selects the counter $c_j$ with the smallest item $j$ (Line~\ref{lss:line_ss_minarg}), the counter of item $i$ will be the counter of $j+1$, and then the counter $j$ is reset (Line~\ref{line:ss_replacing}).
In other words, for $ M \ge 1$, let

\small
\begin{align}
j= \text{argmin}_{j} c_j(u_{1\cdots (x-1)}) \\
c_j(u_{1\cdots (x-1)}) = M \\
f_i(u_{1\cdots (x-1)})  = 0 \\
c_i(u_{1\cdots (x-1)}) = 0
\end{align}
\normalsize

Then after the arrival of $u_x$:

\small
\begin{align}
f_j(u_{1\cdots x}) = f_j(u_{1\cdots (x-1)}) \\
f_i(u_{1\cdots x}) =  1\\
c_j(u_{1\cdots x}) = 0 \label{eq:cj_zero} \\
c_i(u_{1\cdots x}) = M + 1
\end{align}
\normalsize

We now calculate the error of the frequency estimation of items $j$ and $i$. For item $j$, since $M \ge 1$,  Equation~\ref{eq:cj_zero} implies 

\begin{multline*}    
\delta_j(u_{1\cdots (x-1)}) = |f_j(u_{1\cdots (x-1)}) - c_j(u_{1\cdots (x-1)})|  = |f_j(u_{1\cdots (x-1)}) - M| \\ <  |f_j(u_{1\cdots (x-1)})| = |f_j(u_{1\cdots x})|  = |f_j(u_{1\cdots x}) - c_j(u_{1\cdots x})| = \delta_j(u_{1\cdots x})
\end{multline*}

For item $i$:
\begin{multline}    
\delta_i(u_{1\cdots (x-1)}) = |f_i(u_{1\cdots (x-1)}) - c_i(u_{1\cdots (x-1)})| = 0 < |1 - (M+1)| \\ = |f_i(u_{1\cdots x}) - c_i(u_{1\cdots x})| = \delta_j(u_{1\cdots x})
\end{multline}

For all $k \neq i$ and $k \neq j$, since the counter of $k$ is not affected by the arrival of $u_x$
$\delta_k(u_{1\cdots (x-1)}) = \delta_k(u_{1\cdots x})$. Thus,
$\delta(u_{1\cdots (x-1)}) < \delta(u_{1\cdots x})$.

\textbf{Induction step for $t > 0$.}
We start with the induction hypothesis
$\delta(u_{1\cdots (x-1)}v_{1 \cdots (t-1)}) < \delta(u_{1\cdots x}v_{1 \cdots (t-1)})$, and consider the arrival $v_t$.
\paragraph{Case 1: $c_{v_t}(u_{1\cdots x}v_{1 \cdots (t-1)}) > 0$ :} That means Space Saving keeps track of item $v_t$ frequency in counter $c_{v_t}$ over the stream $u_{1\cdots x}v_{1 \cdots (t-1)}$, which implies that  $c_{v_t}(u_{1\cdots (x-1)}v_{1 \cdots (t-1)}) > 0$ because eliminating $u_x$ does not affect the counter of $v_j$.
In both cases, Space Saving updates this counter, which leaves the error unchanged.

\paragraph{Case 2: $c_{v_t}(u_{1\cdots x}v_{1 \cdots (t-1)}) = 0$:} If $v_t$ is a single occurrence item, the analysis follows the base case.

If $v_t$ is not a single occurrence item, if $c_{v_t}(u_{1\cdots (x-1)}v_{1 \cdots (t-1)}) > 0$, it means that the arrival of the single occurrence $u_x$ reset the counter of $v_t$.
Under the stream $u_{1\cdots (x-1)}v_{1 \cdots (t-1)}$, the Space Saving continues to update the counter of $v_t$ without adding any error.
A combination of this and the induction hypothesis yields the desired result.

If $v_t$ is not a single occurrence item,
if $c_{v_t}(u_{1\cdots (x-1)}v_{1 \cdots (t-1)}) = 0$, by the induction hypothesis, both counter vectors have the same set of non-zero entries.
If the Space Saving has fewer than $k$ entries (Line~\ref{line:ss_addentry}), then the algorithms add a counter for $v_t$ and the analysis follows Case 1 above.
Otherwise, Space Saving will pick the counter $c_j$ to replace.
Let
\small
\begin{align}
j= \text{argmin}_{j} c_j(u_{1\cdots x}v_{1\cdots y-1}) \\
j'= \text{argmin}_{j'} c_j(u_{1\cdots (x-1)}v_{1\cdots y-1})
\end{align}
\normalsize
Either $j'=j$ or $j=i$.
In case of $j'=j$: streams ($u_{1\cdots x}v_{1\cdots y}$ and $u_{1\cdots (x-1)}v_{1\cdots y}$) result in updating the counters by the same differences. Combining that with the induction hypothesis yields the result.

In the case of $j=i$: after the arrival of $v_t$:

\small
\begin{align}
c_i(u_{1\cdots x}v_{1\cdots y}) = c_{j'}(u_{1\cdots (x-1)}v_{1\cdots y}) = 0 \\
c_i(u_{1\cdots x}v_{1\cdots (t-1)}) \ge c_{j'}(u_{1\cdots (x-1)}v_{1\cdots (t-1)}) \label{eq:increment_i}
\end{align}
\normalsize
The equation \eqref{eq:increment_i} is based on the fact that each item ($i$ in particular) increases a counter in Space Saving.

The error of the frequency estimation of items $v_t$:

\begin{multline*}    
\delta_{v_t}(u_{1\cdots x}v_{1\cdots y}) = |f_{v_t}(u_{1\cdots x}v_{1\cdots y}) - (c_{i}(u_{1\cdots x}v_{1\cdots y}) +1)| \leq \\  |f_{v_t}(u_{1\cdots x}v_{1\cdots y}) - (c_{j'}(u_{1\cdots (x-1)}v_{1\cdots y}) + 1)| = \delta_{v_t}(u_{1\cdots (x-1)}v_{1\cdots y})
\end{multline*}
\end{proof}

\else
\begin{comment}
We prove Theorem~\ref{theorem:remove_singles} by induction on $y$. Due to lack of space, the proof is deferred to the full version of the paper.
\end{comment}
\fi

%{\bf MM: $t$ appears to be used 2 different ways, time and occurrences, confusing.}
%\textbf{\lsslf{}.}
%
We accordingly introduce \lsslf{} (\lsslfFULL{}), which aims to exclude items that are predicted to have up to $t$ occurrences from being inserted into Space Saving. We refer to the special case when $t=1$ by \lsslfs{} (\lsslfsFULL{}).

\subsection{Addressing Items with Single Occurrence}

\ifdefined\arXiv
\begin{figure}
    \centering
    \includegraphics[width = 8cm]{figs/draws-2.pdf}
    \caption{Block-diagram representation of \lsslf{} and \lsslf{}}
    \label{fig:block_diagram}
\end{figure}
\fi

\begin{figure}[t]
  \centering
  \includegraphics[width=0.5\linewidth]{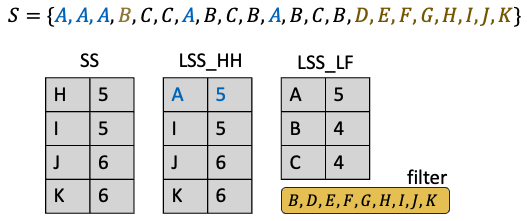}
  \caption{An illustration of our algorithms. Consider an input stream $S$ of 11 items as shown, where $A$ is predicted as a heavy hitter, the first arrival of $B$ is predicted (incorrectly) as low frequency, $D, E, F, G, H, I, J, K$ are predicted (correctly) as low-frequency items. In this example, low-frequency items take over the SS counters. In \lssasn{}, by allocating a fixed entry for the predicted heavy hitter A, the remaining counters again are taken over by low-frequency items. \lsslf{}, however, ensures that low-frequency items do not dominate by filtering them. When $B$ arrives after being previously stored in the filter, it is tracked again. Note that \lsslf{} uses fewer counters than SS and \lssasn{} to account for the additional memory required for the filter. }
  \label{fig:example}
\end{figure}

%With the intuition that removing single occurrence items improves accuracy for Space Saving (Theorem~\ref{theorem:remove_singles}), we consider methods for removing such items.
We focus on this particular case for two primary reasons: First, for some distributions (see Figure \ref{fig:dataset_prediction_hits}), the majority of low-frequency items are items that occur only once. Second, in certain setups, obtaining a predictor for single occurrence items is more achievable than a general low-frequency predictor. Even by addressing this special case of single occurrences, we demonstrate improved performance.

\lsslfs{} employs a predictor that, for a specific item $i$, predicts if item $i$ has a single occurrence. Note that we perform predictions for every incoming arrival.

\begin{comment}
{\bf MM:  We've described robustness now earlier;  the first few stentences here can be removed or merged with that section.  Probably can jump righ to "Unless some mitigating structure..."}
However, machine learning methods are inherently imperfect and may exhibit errors, including substantial and unexpected errors. Therefore, as suggested in the algorithms with predictions literature (see, e.g.,~\cite{mitzenmacher2022algorithms}), algorithms using predictions should demonstrate sufficient robustness to handle prediction errors that may occur. In particular, the notion of robustness that has become common 
is that the performance of algorithms using predictions should not be significantly inferior to that of conventional online algorithms that do not rely on predictions, even if predictions are inaccurate.
\end{comment}

Unless some mitigating structure is added, ignoring predicted single items can lead to unbounded errors.
For example, if a heavy hitter is predicted incorrectly as a single item, this item is excluded from the Space Saving, violating the error guarantee (Lemma~\ref{lemma4}).

We therefore add a structure to ensure that when the predictor suggests an item be ignored, a verification process is used to determine if the item has been previously ignored, and ignores the prediction if that has occurred.  Our approach uses a Bloom Filter (BF)~\cite{bloom1970space, broder2004network}\footnote{The BF could be replaced with any similar filter structure, such as a cuckoo filter or ribbon filter~\cite{dillinger2021ribbon,fan2014cuckoo}.} (we consider further generalizations later).
A BF is a simple and efficient probabilistic data structure that determines whether an item is part of a set, which does not yield false negatives but can yield false positives.

Specifically, we keep a Bloom filter of predicted single occurrence items previously ignored to ensure robustness.
If the predictor predicts an item is a single occurrence, but the Bloom filter returns that it has been previously ignored, this indicates the item is not a single occurrence item, we disregard the predictor's suggestion and instead insert the item into the Space Saving. Otherwise, if the item is not found in the Bloom filter, we add it to the filter and ignore it.
For example, in Figure~\ref{fig:example}, the items $B, D, E, F, G, H, I, J, K$ are predicted as low-frequency items and are eliminated from being inserted into the table. Instead, they are placed in the filter. $B$ is tracked again as it was incorrectly predicted as a low frequency item. As a result, the counters remain dedicated to tracking non-low-frequency items.

The effect is that faulty predictions will not cause us to ignore an item more than once, because Bloom filters are designed so that there are no false negatives.
However, we may underestimate the count of mispredicted items by 1, because an item is not included in the Space Saving when it is inserted in the Bloom filter. We compensate for this by adding $1$ to the query result of the Space~Saving.

%We emphasize that this approach differs from the strawman approach, where the Bloom filter includes all distinct items within the stream. In \lsslfs{}, which takes advantage of predictions, the filter only contains items that the predictor recommends ignoring.

\subsubsection{Robustness Result}

We first present theoretical results showing that \lsslfs{} is robust, in the sense that it cannot behave too much worse than the corresponding algorithm that does not use predictions (denoted by $SS$).

\begin{theorem}
\label{theorem:lss_correctness}
    Let $SS$ be an algorithm for $(\epsilon - \frac{1}{N})$-Frequency. Then \lsslfs{} (Algorithm~\ref{alg:lss}) solves $\epsilon$-Frequency.
\end{theorem}

\begin{proof}

For any item $i$ and stream size $N$, we have \begin{align} \widehat{f_i} \triangleq SS_{(\epsilon - \frac{1}{N})}.\mbox{{\sc Query$(i)$}} + 1.
\end{align}
That is, our estimate is the query result for $i$ from $SS$, which is an algorithm for
$(\epsilon - \frac{1}{N})$-Frequency, with at most one occurence of $i$ removed from the stream $SS$ processes and an extra count of 1 added back in.  It follows the smallest possible return value is 
$(f_i-1)+1 = f_i,$ and the largest possible return value is 
$\left (f_i + \left (\epsilon - \frac{1}{N}\right )N \right ) + 1 = f_i + \epsilon N,$
proving the claim.

Note that we alternatively could have used a $\epsilon$-Frequency algorithm for $SS$ and not added 1 to the query result, yielding a largest possible return value $\widehat{f_i}$ of $f_i + \epsilon N$, but might yield a smallest value of
$f_i-1$ (as a lower bound). We chose the presentation of Theorem~\ref{theorem:lss_correctness} to maintain the same guarantees as without predictions.
\end{proof}

%We now provide a theorem that gives insight into the possible gains from predictors, by considering if we had perfect predictions.

\begin{theorem}
\label{theorem:lss_spacereduction}
Given $k= \epsilon^{-1}$ available counters and the availability of perfect predictions for items with a single occurrence, let $\ell$ represents the count of such items.
By utilizing this predictor specifically to filter out single-occurrence items, Space Saving can estimate the frequency of any item with additive error less than $(N-\ell)\epsilon$, where $N$ denotes the size of the stream. 
\end{theorem}

\begin{proof}
In the Space Saving algorithm, the minimal counter value is always greater than or equal to the ratio of the number of inserted items to the number of counters. Since we eliminate items with single occurrences, the total number of inserted items becomes $N-\ell$.
By using $\epsilon^{-1}$ counters, we ensure that the minimal counter value is greater than or equal to $(N-\ell)\epsilon$. The rest of the proof follows immediately from Lemma~\ref{lemma2},~\ref{lemma4}.
\end{proof}

Our \lsslfs{} algorithm will, of course, not be as good as Theorem~\ref{theorem:lss_spacereduction} even with perfect predictions, because it does not assume predictions are perfect, and uses a Bloom filter for robustness.  However, the following theorem provides insight into why we expect strong benefits, given suitably good predictors. In stating the theorem, we recall that a Bloom filter has a false positive rate, which corresponds to the probability a non-set item creates a false positive, and further, for any sequence of $M$ Bloom filter queries, the fraction of false positives is at most $(1+\nu)M$ for any constant $\nu$ with high probability (that is, $O(M^{-\alpha})$ for any constant $\alpha$; note the result may require sufficiently large $M$).

\begin{theorem}
\label{theorem:lss_predcitionspacereduction}
    Given $SS$ with $k= \epsilon^{-1}$ to solve the $\epsilon$-Frequency problem and given perfect predictions of whether an item is a single occurrence.
    \lsslfs{} using $k$ counters
    can estimate the frequency of any item with additive error less than $(N-\ell \cdot(1-(1+\nu)fpr)\epsilon$ with high probability, where $fpr$ is the false positive rate of its Bloom filter and $\nu > 0$ can be any suitable constant.
\end{theorem}

\begin{proof}
Based on Algorithm~\ref{alg:lss}, \lsslfs{} uses a predictor to filter out arrivals with a single occurrence.
A perfect predictor is assumed, but without knowing that it is perfect.
As \lsslfs{} utilizes a Bloom filter to handle the predictor's imperfections, which may yield false positives that include unnecessary items in its Space Saving instance.
The expected number of false positives can be expressed as $\ell \cdot fpr$, and with high probability the number of false positives is at most $(1+\nu)\ell \cdot fpr$. Following the same logic as in Theorems~\ref{theorem:lss_spacereduction}, it can be deduced that the total number of inserted items, with high probability, is $N - \ell \cdot (1 - (1+\nu) fpr)$.
\end{proof}

%%%%%%%%%%%%%%%%% CBF

\subsection{Addressing Items Up to t Occurrences}
For given $t$, the previous approach can be generalized by addressing items with up to $t$ occurrences, where $t$ is a threshold that depends on the specific distribution.

\begin{comment}
\textbf{Strawman approach.}
The Counting Bloom Filter (CBF)~\cite{fan2000summary} expands on the Bloom filter capabilities by tracking how many times each item appears in a multi-set. In the strawman approach it could be used to place an item into the Space Saving instance only after the count for the item exceeds a threshold value $t$.
However, CBFs consume more space than the traditional BF. Also, instead of false positives, a CBF can have overestimation errors, which again means an item may enter the Space Saving instance ``earlier'' than expected. 
An additional CBF stage can remove the initial $t$ occurrences of each item (when there is no overestimation error). This would exclude items that appear fewer than $t$ times from being included in the Space Saving structure.
To compensate for this removal, when we query Space Saving to estimate the frequency of an item, we add $t$ to the result.
% This adjustment ensures that the approximate frequency takes into account the occurrences removed at the earlier stage.
\end{comment}

%\textbf{Learning approach.}
We now present \lsslf{}, a generalization of \lsslfs{}.
We use a predictor that, for a specific item $i$ and threshold $t$, predicts if item $i$ has more than $t$ occurrences.
\lsslf{} replaces the standard Bloom filter with a Counting Bloom Filter (CBF)~\cite{fan2000summary} which expands on the Bloom filter capabilities by tracking how many times each item appears in a multi-set.
The CBF tracks the number of times an item is inserted into it, so we can (approximately) count how many times an item that is predicted to have less than $t$ occurrences within the stream, and then we place items into the Space Saving if their CBF count reaches a predefined threshold $t$.
This generalizes the previous algorithm using the Bloom filter which corresponds to the case $t=1$.  
Now, however, our underestimation might be as much as $t$, from $t$ insertions until reaching the CBF threshold.  We correct this by adding $t$ to the value returned by $SS$ in \lsslf{}.

%\subsubsection{Robustness Result}

\begin{theorem}
\label{theorem:allow_cbf_correctness}
    Given a threshold $t$, let $SS$ be an algorithm for $(\epsilon - \frac{t}{N})$-Frequency. Then \lsslf{} solves $\epsilon$-Frequency.
\end{theorem}

\begin{proof}
The proof follows the same logic as that of Theorem~\ref{theorem:lss_correctness}.
Here, however, we may lose up to $t$ entries for an item $i$ within the stream.
\end{proof}

\section{\lssasn{}: \lssaeFULL{}}
\label{sec:lss_hh}

\ifdefined\arXiv
\begin{figure}
    \centering
    \includegraphics[width = 8cm]{figs/lss_assigned.pdf}
    \caption{Block-diagram representation of \lssasn{}}
    \label{fig:lss_asn_block_diagram}
\end{figure}
\fi

Here, we explore the case where a predictor for heavy hitters is included.
We propose \lssasn{} as an approach that exploits the heavy hitter predictor for better performance.
\lssasn{} assigns a set number of entries specifically for predicted heavy hitters in the Space Saving algorithm. As a result, these entries are protected from replacement, even when a new item appears. The other remain mutable, as with the original Space Saving algorithm;  in particular, when an item appears that replaces the item with the lowest counter value, it does so only with respect to the mutable entries.  We discuss additions to this approach to ensure it remains robust in the face of potential prediction errors.

\begin{comment}
{\bf MM: It is strange be talking about training here -- makes it sounds like you train and pick specific items, rather than place predicted heavy hitters.  I'd remove the paragraph below;  seems redundant with the paragrph above, whic I'm addin to.}
We first train a predictor to identify heavy hitters and allocate a fixed number of entries to these heavy hitters within the Space Saving table. That is, we prevent the removal of these particular items, even when they display the lowest counter value — a scenario where they would ordinarily be replaced in the traditional Space Saving. In such cases, \lssasn{} eliminates a mutable entry with the smallest counter.
Therefore, we preserve historical information stored in heavy hitter counters, which may be lost if a small item triggers the replacement of a heavy hitter.
\end{comment}

Inaccurate predictions can lead to a scenario where small items erroneously take up fixed entries and reduce the available fixed entries originally intended for actual heavy hitters.
As a result, heavy hitters are subject to frequent inclusion and removal from the structure, resulting in unbounded errors.
To ensure the robustness of our approach, we set a limit on the maximum number of fixed entries ($k_{hh}$).
In cases where the number of predicted heavy hitters exceeds the defined threshold, the initial heavy hitters encountered in the data stream fill the fixed entries, while the remaining heavy hitters compete on regular mutable entries.
In contrast, if the number of predicted heavy hitters is less than the fixed entries, those entries will be given to non-heavy hitters. This is due to the fact that we only mark entries as fixed when a heavy hitter shows up.
In Figure~\ref{fig:example}, there is one fixed entry allocated for $A$, which is predicted as a heavy hitter. This prevents low-frequency items from replacing the counter tracking the heavy hitter $A$.

%Algorithm~\ref{alg:lssasn} provides pseudo-code where the additions compared to the traditional Space Saving are highlighted in \textcolor{blue}{blue}.
% {\bf MM:  we might want to mention that one could randomly allow entry to the fixed cells, using a hash function;  we are implicitly assuming the first predicted heavy hitters are no more likely to be mispredictions than later ones, which seems a reasonable assumption, but we could modify the sturcture.}
% \rana{why do we use a hash function? do you mean to pick a random entry? another option is to fix the entry only when the heavy hitter exists in the table}

%\rana{another way to ensure robustness is to assign entry for a heavy hitter in probability $\delta$ -- the heavy hitter must show "enough time" before entry is fixed. Here the probability increases robustness unlike the direction in LSS+}

\begin{theorem}
\label{theorem:lss_predcitionspacereduction}
    Given the availability of perfect predictions for heavy hitters, 
    \lssasn{} using $k$ counters, where $k_{hh}$ among them are fixed entries, 
    provides exact frequencies for $k_{hh}$ heavy hitters (zero errors) and estimates the frequency of the other items with additive error less than $\frac{N-k_{hh}\theta N}{k-k_{hh}}$ where $\theta$ is the heavy hitters' threshold.
\end{theorem}

\begin{proof}
Since heavy hitters have dedicated counters, their error is zero. For the remaining items, we have $k-k_{hh}$ (mutable) counters. As before, the minimal counter value among the mutable counters is always greater than or equal to the ratio of the number of inserted items to the number of counters. However, since we eliminate heavy hitter items, each appearing at least $\theta N$ times, the total number of inserted items becomes $N-k_{hh}\theta N$. Thus, the rest of the proof follows immediately from Lemma~\ref{lemma2} and Lemma~\ref{lemma4}.
\end{proof}

If the frequency distribution of items is Zipfian, then the number of $\theta$-heavy hitters is at most $\frac{1}{\theta \ln{n}}$, where $n$ represents the total number of unique items in the stream (Remark 9.13 in~\cite{hsu2019learning}). In this case, $k_{hh}$ can be set to $\frac{1}{\theta \ln{n}}$.

%===========================

%\rana{theorem on error guarantee of hh estimation}

%\rana{theorem on error guarantee of any item frequency estimation when using two-sided error}

\section{\lssplus{}: Faster \lss{}}
\label{sec:lss_plus}

%%%%%%%%%%%%%%%%% LSS+
%\ifdefined\EXTENDED

\begin{algorithm}[t]
\caption{\lssplus{}}  \label{alg:lssplus}
\begin{algorithmic}[1]
\State \textbf{Initialization}: $T \leftarrow \varnothing$,
\State $BF$ -- $Bloom filter$

    \Function{Insert}{$i$}
    \If{item $i$ is not predicted as low-frequency}
        \BoxedState{Learned Space Saving Insertion}
    \Else
        \If{$Uniform(0, 1) \leq \tau$}\Comment{With probability $\tau$}  \label{line:opt1}
        \If{$BF.CONTAINS(i)$}
        \If{$i \in T$}
            \State $c_i \leftarrow c_i + \tau^{-1}$
        \ElsIf{$|T| < k$} \label{line:ssplus_addentry}
            \State $T \leftarrow T \cup \{i\}$
            \State $c_i \leftarrow \tau^{-1}$
        \Else
            \State $j \leftarrow \text{argmin}_{j \in T} c_j$ 
            \State $c_i \leftarrow c_j + \tau^{-1}$ 
            \State {$T \leftarrow T \cup \{i\} \setminus \{j\}$}
        \EndIf
        \Else
        \State $BF.ADD(i)$
        \EndIf
        \EndIf
    \EndIf
    \EndFunction

    \Function{Query}{$i$}
        \BoxedState{$result =$ Learned Space Saving Query}
        \State \Return $result + \tau^{-1}$
    \EndFunction
\end{algorithmic}
\end{algorithm}
%\fi

The inclusion of the filter (BF or CBF), while crucial for robustness, can potentially slow down overall performance and increase memory usage. 
We present \lssplus{} that offers potentials for speedup at the cost of loosening our deterministic approximation guarantees to probabilistic guarantees. For simplicity, we describe \lssplus{} for the case of filtering out single occurrences, but the approach also applies to filtering out items with up to $t$ occurrences.

Whenever the predictor suggests an item will have a single occurrence, \lss{} checks for the item's presence in the Bloom filter, and either adds it to the Bloom filter if it is not there, or adds 1 to the item count in the Space Saving data structure if it is. \lssplus{} mitigates this overhead by selectively executing this step with probability $\tau$, and adding $\tau^{-1}$ to the item count in the Space Saving data structure if needed.  (For convenience, we assume here that $\tau^{-1}$ is an integer, to avoid floating point arithmetic, but one could develop alternative implementations.)  

In the special case where $\tau$ equals $1$, \lssplus{} coincides with \lss{}. When $\tau$ is less than 1, it takes on average $\tau^{-1}$ predictions that an item will have a single occurrence before checking the Bloom filter, reducing accesses to and memory required for the Bloom filter.  However, now the guarantees for the algorithm are only in expectation; when an item is in the Bloom filter, its {\em expected} count increases by one each time such a prediction occurs.  The value of the parameter $\tau$ in \lssplus{} can be chosen to balance between robustness and efficiency.
This optimization is handled in Line~\ref{line:opt1} of Algorithm~\ref{alg:lssplus}.

As \lssplus{} checks the Bloom filter probabilistically, it makes sense to consider the expected value of the query response;  we refer to this as solving the query problem in expectation.

\begin{theorem}
\label{theorem:optallow_correctness}
    Let $SS$ be an algorithm for $(\epsilon - \frac{\tau^{-1}}{N})$-Frequency. Then \lssplus{} solves $\epsilon$-Frequency in expectation.
\end{theorem}

\begin{proof}
The proof follows similar logic as that of Theorem~\ref{theorem:lss_correctness}.
The return value of \lssplus{} is \begin{align} \widehat{f_i} \triangleq SS_{(\epsilon - \frac{\tau^{-1}}{N})}.\mbox{{\sc Query$(i)$}} + \tau^{-1}.
\end{align}

%We explain the addition of $\tau^{-1}$.
Each time item $i$ is encountered, if the prediction says it will appear again in the stream, its count is increased by 1 deterministically.  If the prediction says it has a single occurrence, but it is in the Bloom filter, its expected count increases by 1, as it goes up by $\tau^{-1}$ with probability $\tau$ (when the Bloom filter is checked).  The only times the expected count does not increase by 1 on each appearance is when the item is not in the Bloom filter and it is predicted to be a single occurrence; this occurs an expected $\tau^{-1}$ times before the item is put in the Bloom filter.  It follows that  
$E[\widehat{f_i}] \geq f_i$ and
%\vspace{-0.05in}
$E[\widehat{f_i}] \leq f_i + \left (\epsilon - \frac{\tau^{-1}}{N} \right )N + \tau^{-1} = f_i + \epsilon N.$

\end{proof}

\section{Evaluation}
\label{sec:eval}

%\begin{comment}
Our evaluation examines the effectiveness of \lss{} in comparison to SS. To this end, we aim to address the following questions:
\begin{itemize}
    \item 
How robust is LSS to prediction errors? We consider two extreme cases: all items are predicted as low frequency, and all items are predicted as heavy hitters. We also examine the effect of prediction error probability on \lss{} accuracy.
    \item 
How does \lss{} accuracy compare against SS on various applications, such as frequency estimation, identifying top-k items, and finding heavy hitters?
    \item 
How should we adjust parameters (e.g. the number of fixed entries, filter size) depending on the specific task: frequency estimation, identifying top-k items, and identifying heavy hitters)?
    \item
How much does \lssplus{} improve the update throughput compared to \lss{}?
\end{itemize}

\begin{figure}[t]
	\center{
		\begin{tabular}{ccc}
			\subfloat[Web search]{\includegraphics[width=0.3\columnwidth]{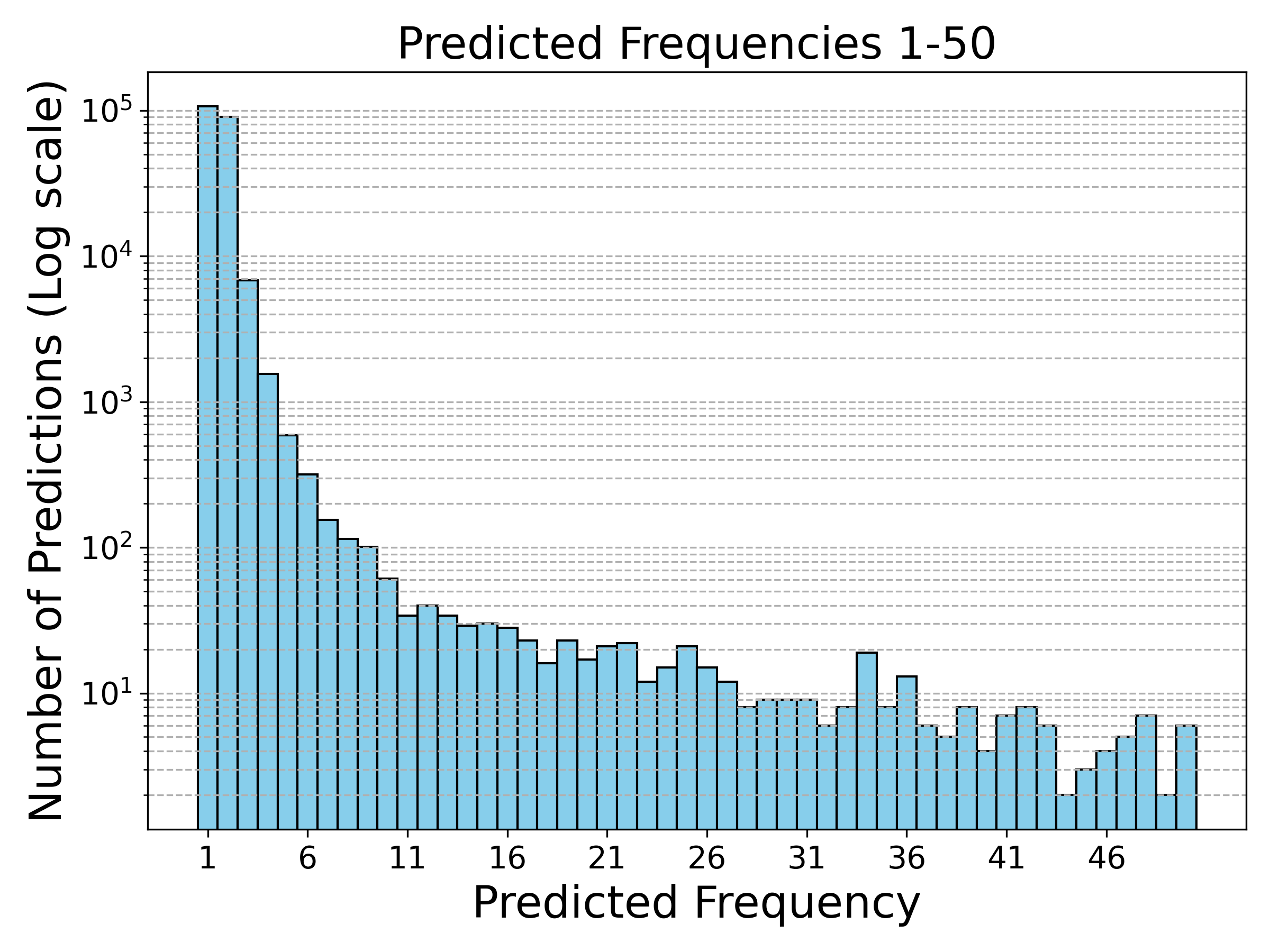}\label{fig:dist_aol}} &
			\subfloat[IP]{\includegraphics[width=0.3\columnwidth]{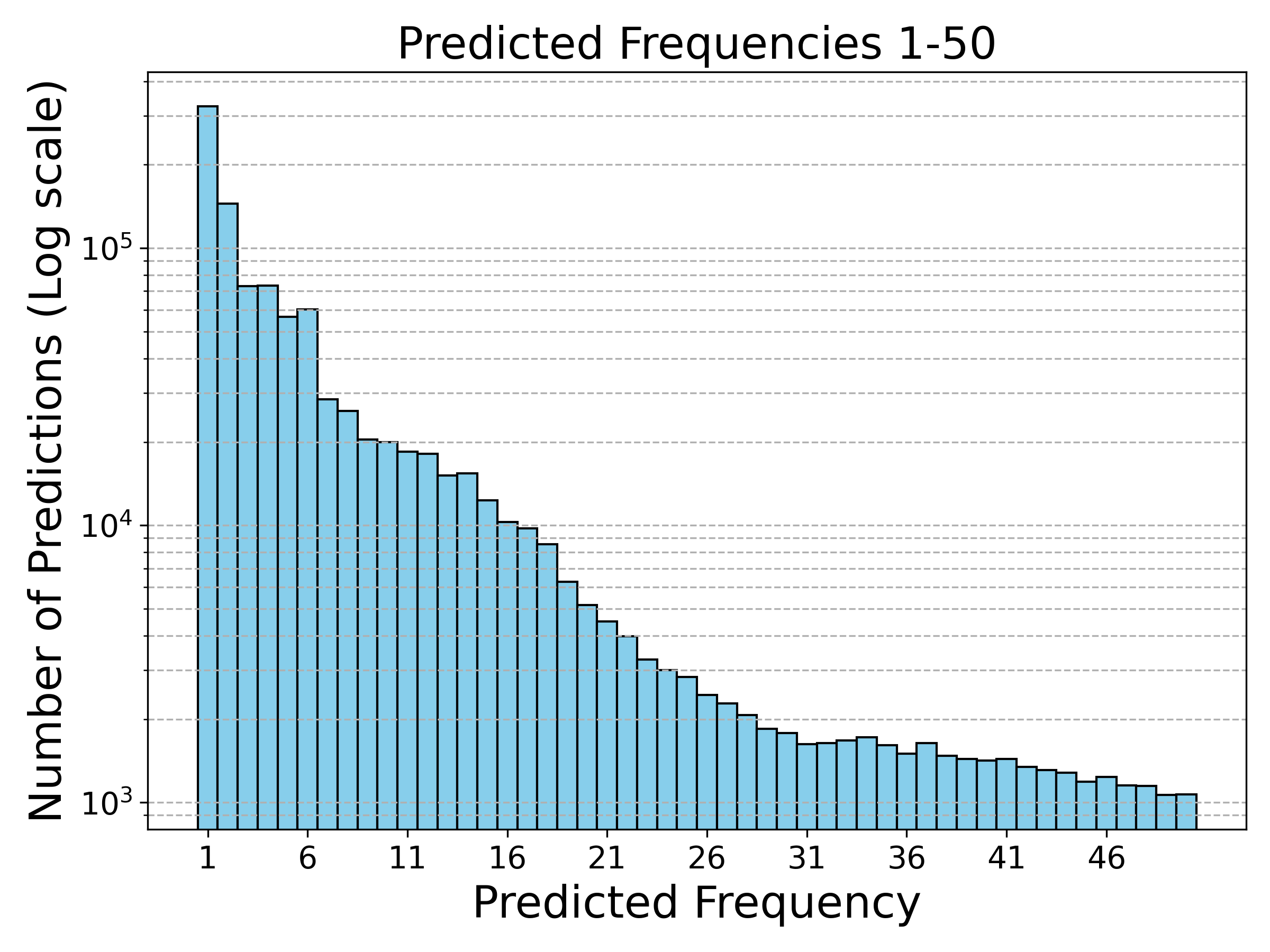}\label{fig:dist_ip}}
			\subfloat[Zipf $\alpha=1.3$]{\includegraphics[width=0.3\columnwidth]{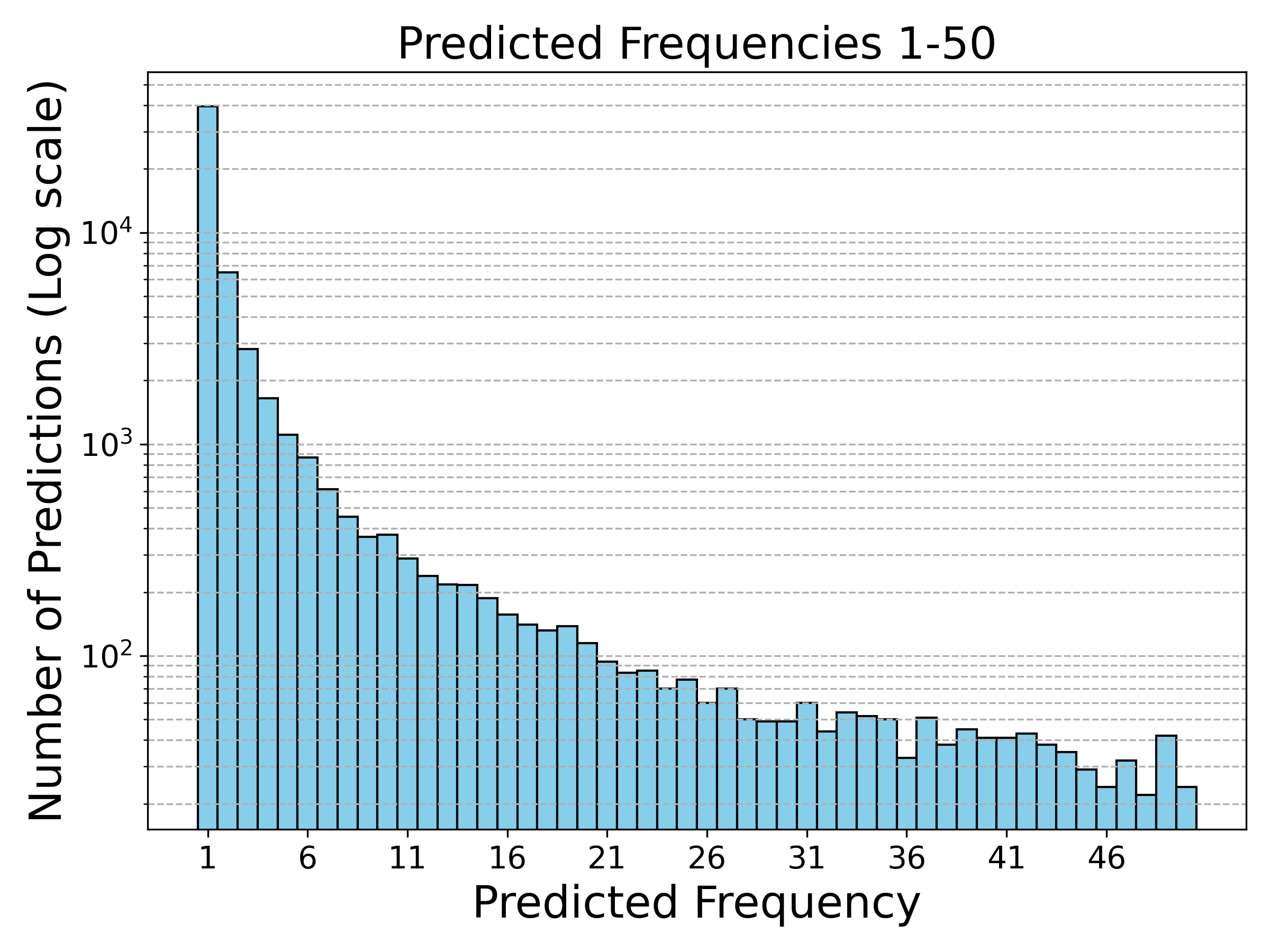}\label{fig:dist_synthetic}}
		\end{tabular}
		}
	\caption{Histogram of predicted frequencies up to $50$ of the used datasets (Log scale). For Web search and IP datasets, we use the learned model described in Section~\ref{sec:learned_predictor}, for the Zipf distribution, we used the simulated predictor with $p=0.9$.}
	\label{fig:dataset_prediction_hits}
\end{figure}

%\subsection{Experimental Setup}
%{\bf MM: I think what is now section 7.1 above should be a subsection of this experimental setup;  for example, 7.1 acts like we know the datasets, but we don't present them until here!}

\textbf{Datasets:}
\begin{itemize}
    \item \textbf{IP Trace Datasets:} :
    We use the anonymized IP trace streams collected from CAIDA~\cite{CAIDACH16}. The traffic data is collected at a backbone link of a Tier 1 ISP between Chicago and Seattle in 2016. Each recording session lasts approximately one hour, with around 30 million packets and 1 million unique flows observed within each minute.
    
    \item \textbf{Web Search Query Datasets:} We use the AOL query log dataset~\cite{AOL}, which comprises 21 million search queries collected from 650 thousand anonymized users over a 90-day period. The dataset contains 3.8 million unique queries, each consisting of a multi-word search phrase.

    \item \textbf{Synthetic Datasets:} We generated synthetic datasets following the Zipf~\cite{powers1998applications} distribution with varying skewness levels, each dataset contains 10 million items.
\end{itemize}

\textbf{Implementation and Computation Platform:}
We implemented the learned versions of Space Saving in Python 3.7.6. The evaluation was performed on an AMD EPYC 7313 16-Core Processor with an NVIDIA A100 80GB PCIe GPU, running Ubuntu 20.04.6 LTS with Linux kernel 5.4.0-172-generic, and TensorFlow 2.4.1.

\textbf{The Learned Model:}
\label{sec:learned_predictor}
We follow the implementation of the learned model in~\cite{hsu2019learning} and adapt it using TensorFlow 2.4.1 for the discussed datasets CAIDA~\cite{CAIDACH16} and AOL~\cite{AOL}. For the CAIDA dataset, we train a neural network to predict the log of the packet counts for each flow. The model takes as input the IP addresses, ports, and protocol type of each packet. We employ Recurrent Neural Networks (RNNs) with 64 hidden units to encode IP addresses and extract the final states as features. Ports are encoded by two-layer fully-connected networks with 16 and 8 hidden units. The encoded IP and port vectors are concatenated with the protocol type, and this combined feature vector is used for frequency estimation via a two-layer fully-connected network with 32 hidden units.
For the AOL dataset, we construct the predictor by training a neural network to predict the number of times a search phrase appears. To process the search phrase, we train an RNN with LSTM cells that take the characters of a search phrase as input. The final states encoded by the RNN are fed to a fully connected layer to predict the query frequency.
The model is trained on a subset of data to identify properties that correlate with item frequencies instead of memorizing specific items. This trained model is then tested on a separate dataset.

\textbf{Simulated Frequency Prediction:}
For synthetic datasets, given true frequencies, we generate predicted frequencies.% with each occurrence of an item.% receiving a unique prediction based on the simulation parameters at the time of prediction.
We simulate a predictor as follows: given a threshold $t$ and probability $p$, items are classified as either small (true count $< t$) or big (true count $\ge t$). With probability $1-p$, an item is mispredicted where small (big) items are mispredicted as big (small) items, and their predicted count is randomly chosen from the set of big (small) item counts. With probability $p$, the prediction of an item is its true count multiplied by a factor that slightly varies around $1$, within a range defined by the noise level (default $5\%$).%, which ensures that noise-adjusted counts remain proportionate to their original values.
In addition, for items below the threshold $t$, a small probability (default $1\%$) controls the likelihood that the added noise will cause these items to be predicted above $t$. For CAIDA and AOL datasets, we use the learned predictor from Section~\ref{sec:learned_predictor}. When demonstrating robustness on real datasets, we employ controlled prediction with a simulated predictor to introduce non-perfect predictions.

% For the CAIDA and AOL datasets, the default model used is the learned predictor explained in Section~\ref{sec:learned_predictor}. However, when demonstrating the robustness of our approach using real datasets, we employ controlled prediction using a simulated predictor. Thus, all experiments have non-perfect predictions.

\textbf{Parameter Setting:} 
Our approach does not aim to optimize every parameter thoroughly. In practice, parameters can be fine-tuned based on knowledge of data distribution, or by iterative refinement over time, benefiting from our scheme's inherent robustness. Here we state the default parameters we used in the experiments unless stated otherwise. For a fair comparison with the same memory consumption as SS, we allocated $90\%$ of the memory consumed by SS counters to the \lss{} counters and the remaining $10\%$ to the filter (CBF), we refer to this as filter ratio in the experiments. When using fixed counters, particularly for finding heavy hitters, we designated $10\%$ of the total counters as fixed counters.
%{\bf MM: clarify what this means;  I think you mean we also do an experiment varying the number of fixed counters.}
(We also explore the effects of varying the number of fixed counters.) Since we do not assume a prior knowledge of the dataset, we set a low frequency threshold $t=4$. 
%{\bf MM: Is the threshold again the default, but we do an experiment varying it?} \rana{TODO this graph}
To set the threshold for identifying heavy hitters, we used the theoretical error guarantee $\epsilon$ based on the used memory (Lemma~\ref{lemma2}). The default value for this threshold is chosen to be $\theta=0.25\epsilon$.
To simulate periodic queries throughout the data stream, we execute a query at intervals of every $1000$ arrivals.
%To measure accuracy across the data stream, we run a query every 1000 arrivals, {\bf MM:  can you specify what you mean by "considering" here?  These different problems have different measurements.  Maybe wait and explain for each.}
%considering all items that have appeared so far. 
The reported accuracy is the average across all windows. 

\subsection{Problems and Metrics}
We explored three problems: (1) detecting heavy hitters in the stream, i.e., items whose frequency exceeds a given threshold; (2) finding the top $k$ most frequent items in the stream, where $k$ is given; and (3) estimating the frequencies of individual items in the stream. Our error metrics are:

\textbf{Root Mean Square Error (RMSE) for Frequency Estimation}: measures the square root of the average squared differences between the estimated frequency and actual frequency. $\mathrm{RMSE} = \sqrt{\frac{1}{n}\sum_{i=1}^{n}(f_i - \hat{f}_i)^2}$.

\textbf{Precision for Top-k}: 
ratio of the number of correctly reported instances to the number of reported instances ($\frac{TP}{TP+FP}$), where $TP$ is the true positive and $FP$ is the false positive.

\textbf{Recall for Heavy Hitters}: 
ratio of the number of correctly reported instances to the number of correct instances ($\frac{TP}{TP+FN}$) where $FN$ is the false negative.

\textbf{Operations performance:} insertions or queries per second.

\subsection{End-to-End Performance}

\paragraph{Data Skew}
To show how the numerous less frequent tail items could collectively dominate the counters in a space-saving algorithm, we present in Figure~\ref{fig:dataset_prediction_hits} histograms for the frequency range of 1 to 50, plotted on a logarithmic scale. The scale highlights that there are a large number of low-frequency items. We observe a sharp peak at the lowest frequency (1) and as the predicted frequencies increase, the number of items decreases roughly exponentially.

\paragraph{Predictions overhead}

The inference time of the used predictors is 2.8 microseconds per item on a single GPU without optimization, which ensures minimal impact on throughput. To evaluate memory overhead, we calculated the total number of parameters and corresponding memory requirements for each predictor model, which depends on the dataset. %For the network flow predictor, the model comprises 167,521 parameters: 76,800 for IP address embeddings (RNN with 64 hidden units), 82,176 for encoding ports (two fully-connected layers with 16 and 8 hidden units), and 8,545 for the frequency estimation network (two fully-connected layers with 32 hidden units). With 32-bit float representations, this translates to approximately 0.67 MB of memory usage. For the AOL dataset predictor, the model has 7,680 parameters: 4,000 for character embeddings (RNN with LSTM cells), and 3,680 for the fully-connected layer. This results in approximately 0.03 MB of memory usage.
For the network flow predictor, the model has 167,521 parameters: 76,800 for IP address embeddings, 82,176 for encoding ports, and 8,545 for the frequency estimation network. With 32-bit float representations, this translates to around 0.67 MB of memory. The AOL dataset predictor is even more lightweight, with 7,680 parameters: 4,000 for character embeddings (RNN with LSTM cells), and 3,680 for the fully-connected layer, using approximately 0.03 MB of memory.
Thus, the prediction overhead, which includes inference and memory, is therefore small.
Note that inference overhead is expected to be less significant in the future~\cite{kraska2018case} due to specialized hardware such as Google TPUs, hardware accelerators, and network compression~\cite{han2017ese},~\cite{sze2017efficient},\cite{chen2016eyeriss},~\cite{han2016eie}. Furthermore, Nvidia has predicted that GPUs will get 1000x faster by 2025.
In terms of memory, there is a growing research field, TinyML~\cite{dutta2021tinyml}, focused on creating tiny machine learning models for efficient on-device execution. It involves model compression techniques and high-performance system design for efficient ML.
Here we present one example of predictors, which can be treated as black boxes without focusing on their internal functioning; our approach can therefore be used with any suitable and efficient learning scheme that yields a predictor, given the rapid advancements in machine learning research.

\paragraph{Robustness}
Figure~\ref{fig:robutsness} shows the robustness of \lss{} that our theorems suggest using the web search dataset with a simulated predictor. Figure~\ref{fig:robust_all_1} evaluates the precision of top-k as function of the consumed space when $k=10$, where the prediction for every item arrival is 1. Even in this extreme case, we observe only a small degradation in accuracy compared to SS. This degradation arises because we do not benefit from filtering since the filter is overloaded with items, causing items to be frequently replaced in the SS table. Additionally, the filter consumes some memory, so less memory is available for tracking items compared to SS. At the other extreme, Figure~\ref{fig:robust_all_hh} shows the recall of finding heavy hitters when every item is predicted as a heavy hitter. In this scenario, the fixed entries in the counters could be filled with non-heavy hitter items, leaving fewer counters available for tracking the actual heavy hitters. This results in a decreased recall rate. However, once again, the decrease is small. Figure~\ref{fig:acc_vs_p} illustrates the precision rate of the top-k task as a function of the prediction accuracy ($p$), as explained in the simulated predictor. When $p=0$, it implies that all items are mispredicted, in which case SS yields a higher recall rate than LSS because the predictions provide no benefit. As $p$ increases, LSS outperforms SS, and in the other extreme, when $p=1$ (perfect prediction), LSS achieves around $50\%$ improvement in precision compared to SS.

\paragraph{Accuracy vs. Fixed Counters} % for hh

\begin{figure*}[t]
    \centering
        \begin{tabular}{ccccc}
            \subfloat[All predictions are $1$]{\includegraphics[width=\smatrixCellWidth]{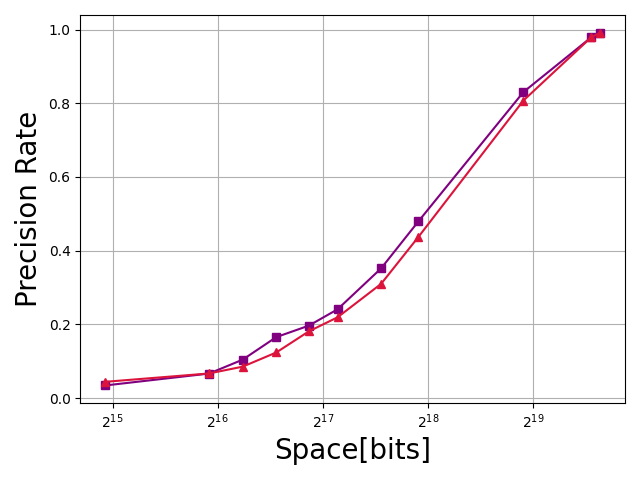}\label{fig:robust_all_1}}&
            \subfloat[All predictions are heavy hitters]{\includegraphics[width=\smatrixCellWidth]{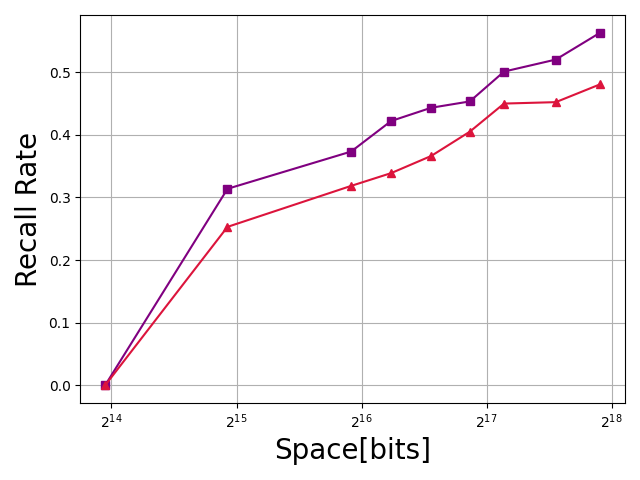}\label{fig:robust_all_hh}} &
            \subfloat[Predictions accuracy]{\includegraphics[width=\smatrixCellWidth]{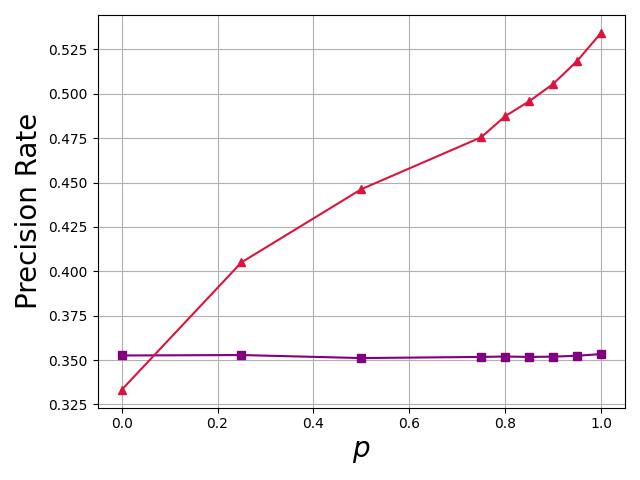}\label{fig:acc_vs_p}} &
            \subfloat[Top-k vs. fixed counters]{\includegraphics[width=\smatrixCellWidth]{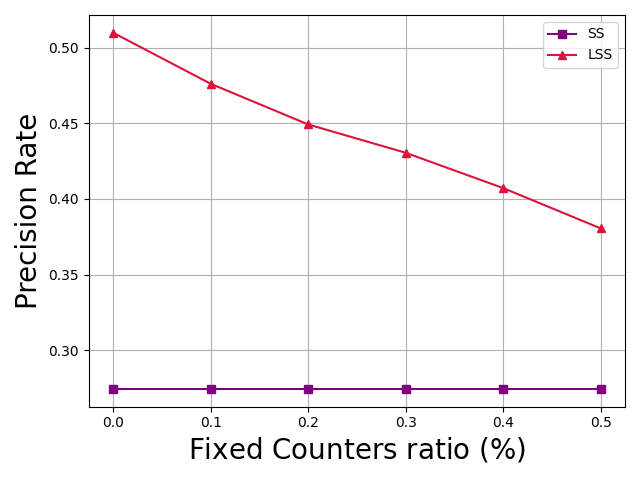}\label{fig:fixed_counters_topk}} &
		\subfloat[Heavy hitters vs. fixed counters]{\includegraphics[width=\smatrixCellWidth]{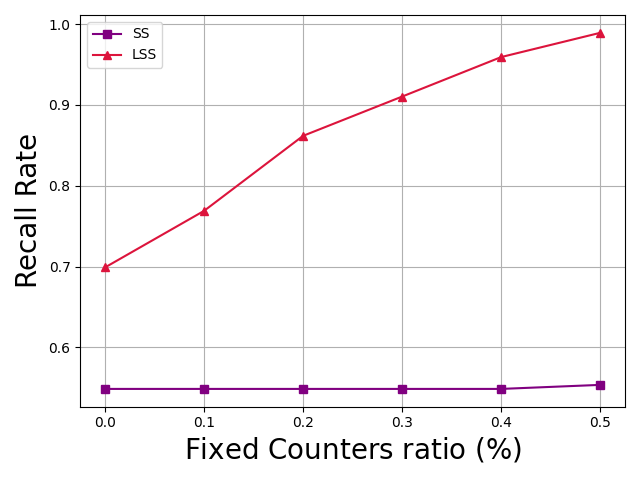}\label{fig:fixed_counters_hh}}
            \\
            \multicolumn{5}{c}{\subfloat{\includegraphics[width=\ssmatrixCellWidth]{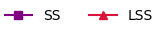}}}
        \end{tabular}
    		
	\caption{(a-c) Robustness of \lss{} using web search dataset (a) precision of top-k ($k=10$) with all predictions as 1 (b) recall of finding heavy hitters when all predictions are heavy hitters (c) precision of top-k vs. prediction accuracy $p$. (d-e) Impact of fixed counters on top-k ($k=64$) and heavy hitters using web search dataset.}
	\label{fig:robutsness}
\end{figure*}

\begin{comment}
    
\begin{figure}[t]
	\centering
	\begin{tabular}{cc}
		\subfloat[Top-k vs. fixed counters]{\includegraphics[width=0.49\columnwidth]{graphs/hh/fixed_counters/aol_topk_precision_vs_fixedCnts_queryOnW.png}\label{fig:fixed_counters_topk}} &
		\subfloat[Heavy hitters vs. fixed counters]{\includegraphics[width=0.49\columnwidth]{graphs/hh/fixed_counters/aol_hh_recall_vs_fixedCnts_queryOnW.png}\label{fig:fixed_counters_hh}}
	\end{tabular}
	\caption{Impact of fixed counters on top-k ($k=64$) and heavy hitters using web search dataset.}
	\label{fig:fixed_counters}
\end{figure}
\end{comment}

\begin{comment}
    
\begin{figure*}[t]
	\center{
		\begin{tabular}{ccc}
			\subfloat[Top-k vs. fixed counters]{\includegraphics[width=\smatrixCellWidth]{graphs/hh/fixed_counters/aol_topk_precision_vs_fixedCnts_queryOnW.png}\label{fig:fixed_counters_topk}} &
			\subfloat[heavy hitters vs. fixed counters]{\includegraphics[width=\smatrixCellWidth]{graphs/hh/fixed_counters/aol_hh_recall_vs_fixedCnts_queryOnW.png}\label{fig:fixed_counters_hh}}
		\end{tabular}
		}
	\caption{Impact of fixed counters on top-k ($k=64$) and heavy hitters using web search dataset.}
	\label{fig:fixed_counters}
\end{figure*}
\end{comment}

Figures~\ref{fig:fixed_counters_topk},~\ref{fig:fixed_counters_hh} show the recall of finding top-k items and heavy hitters as a function of the number of fixed counters using the web search dataset with the learned model. As this number increases, the recall for identifying top-k items decreases because the fixed counters may be populated with heavy hitters that are not among the top-k items. Since these entries are fixed in a first-come manner, the top items may not be placed in the fixed counters. Allocating fixed counters for heavy hitters that are not top-k items results in fewer mutable counters for tracking top-k. However, as the number of fixed counters increases, the recall for detecting heavy hitters improves since having fixed counters dedicated to tracking heavy hitters aligns with this objective. Thus, for finding top-k items, we set the number of fixed entries to zero, while for finding heavy hitters, we allocated $10\%$ of the counters as fixed counters.

%{\bf MM: Not clear in this experiment what the prediction accuracy is in the text.  With perfect prediction, do we see top-k getting worse with fixed counters?  If so I worry there is a bug;  with perfect predctions it should help, right?}

\paragraph{Accuracy vs. Memory}
We examine the accuracy of SS, \lss{}, and \lssplus{} (with $\tau=0.5$) across three tasks: finding top-k items, identifying heavy hitters, and frequency estimation using web search, IP, and synthetic datasets. The accuracy is evaluated as a function of the memory used. Figure~\ref{fig:acc_memory} shows the results for top-k and heavy hitter identification tasks. As expected, higher available memory results in improved precision and recall rates. For the top-k task, no fixed entries were used. \lss{} and \lssplus{} achieve better precision than SS when finding top-k items and better recall when identifying heavy hitters. \lss{} slightly outperforms \lssplus{}, as one might expect;  here for both the filters we allocate
% because \lssplus{} ensures robustness only in probability, 
allocate $10\%$ of the memory. Figure~\ref{fig:acc_memory_frequency} illustrates the RMSE of frequency estimation and shows the accuracy of each variant, \lsslf{} and \lssasn{}, separately. As expected, increasing memory consumption decreases RMSE. We observe that each variant improves the accuracy compared to SS, and the combined usage of techniques in 
\lss{} achieves higher accuracy in frequency estimation.

\begin{comment}
Our evaluation begins with a comparison of \lss{} and the baseline strawman approach.

Figure~\ref{fig:lss_vs_strawman} shows the Root Mean Square Error (RMSE) as a function of memory.
We find the strawman approach requires more memory for the same RMSE, as expected, because its Bloom filter includes all unique items of the dataset.
Figure~\ref{fig:lss_acc_mem} illustrates the RMSE as a function of the memory requirements for \lss{}, \lssplus{} (configured with $\tau = 0.5$), \lsscbf{} (configured with $t = 4$) and SS.
As expected, as memory consumption increases, RMSE decreases.
Comparing the four algorithms, we found that they differed in their memory requirements, and as noted previously they differ in the complexity of their respective predictors.
Specifically, \lsscbf{} outperforms \lss{} and \lssplus{} in terms of accuracy, but requires a more complex predictor.
For \lss{} compared to \lssplus{}, \lssplus{} uses less memory because fewer items are added to the filter, but for \lssplus{} the query result is adjusted by $\tau^{-1}$ (Line~\ref{lssplus:line_returncnt} in Algorithm~\ref{alg:lssplus}). We find as a result the RMSE for \lss{} and \lssplus{} are approximately the same.
Most importantly, the learned versions, improve accuracy over SS.
\end{comment}

\begin{figure*}[t]
    \centering
    \begin{tabular}{cccccc}
        \subfloat[Web, Top-k]{\includegraphics[width=\ssmatrixCellWidth]{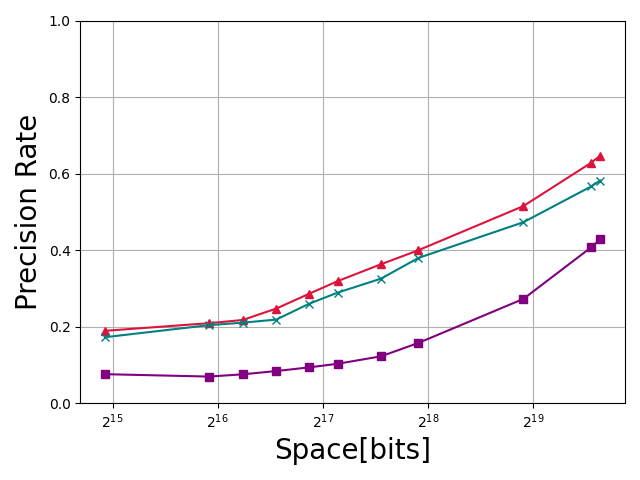}\label{fig:topk_memory_web}} &
        \subfloat[Web, HH]{\includegraphics[width=\ssmatrixCellWidth]{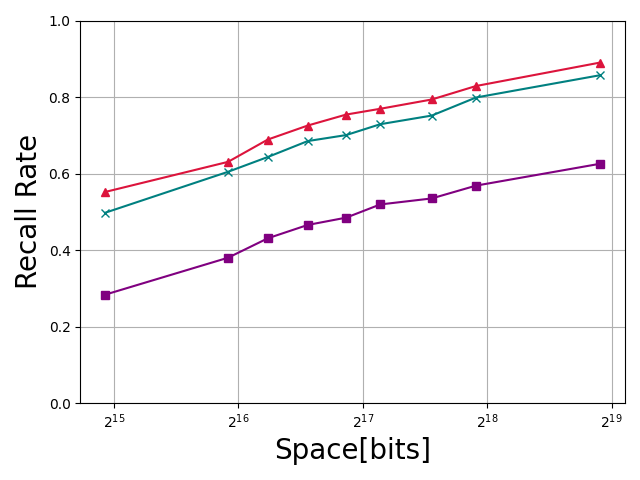}\label{fig:hh_memory_web}} &
        \subfloat[IP, Top-k]{\includegraphics[width=\ssmatrixCellWidth]{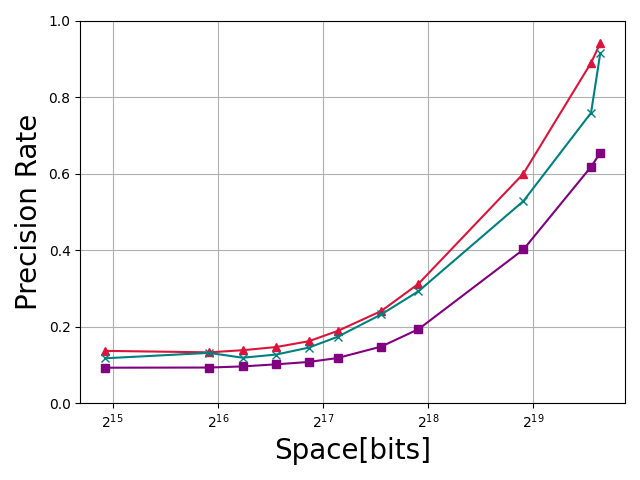}} &
        \subfloat[IP, HH]{\includegraphics[width=\ssmatrixCellWidth]{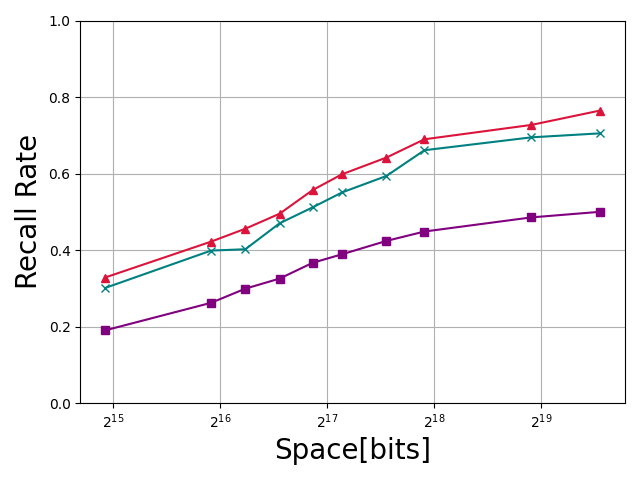}\label{fig:hh_memory_ip}} &
        \subfloat[Zipf, Top-k]{\includegraphics[width=\ssmatrixCellWidth]{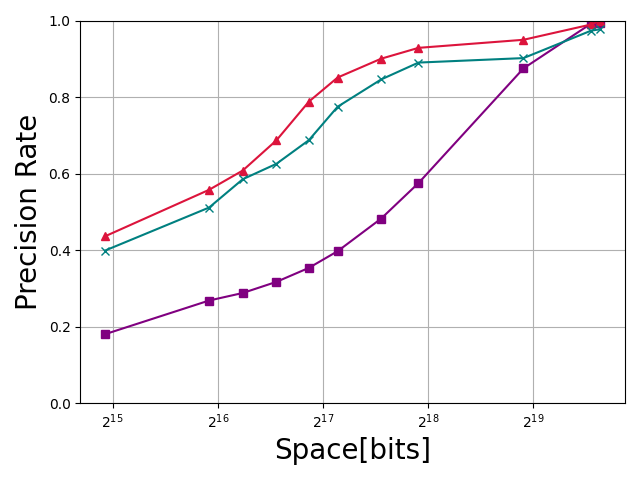}\label{fig:topk_memory_zipf}} &
        \subfloat[Zipf, HH]{\includegraphics[width=\ssmatrixCellWidth]{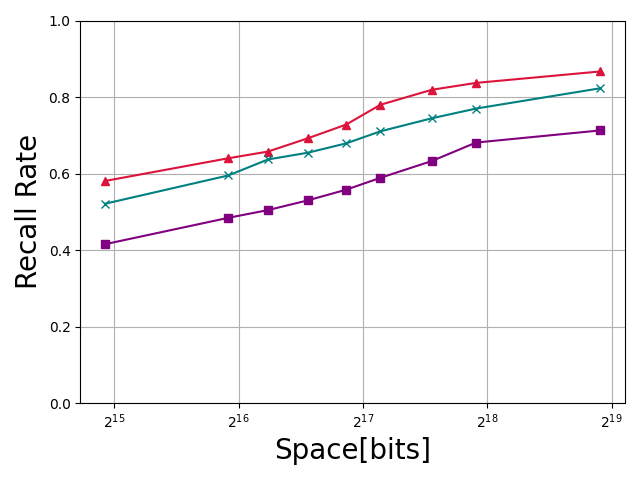}\label{fig:hh_memory_zipf}} \\
        \multicolumn{6}{c}{\subfloat{\includegraphics[width=\matrixCellWidth]{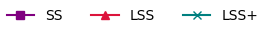}}}
    \end{tabular}
    \caption{Precision and recall vs. memory of identifying top-k ($k=64$) and heavy hitters. \lssplus{} configured with $\tau= 0.5$. We use web search, IP and Zipf ($\alpha=1.3$) datasets.}
    \label{fig:acc_memory}
\end{figure*}

\begin{figure*}[t]
    \centering
    \begin{tabular}{ccccc}
        \subfloat[Web search]{\includegraphics[width=\smatrixCellWidth]{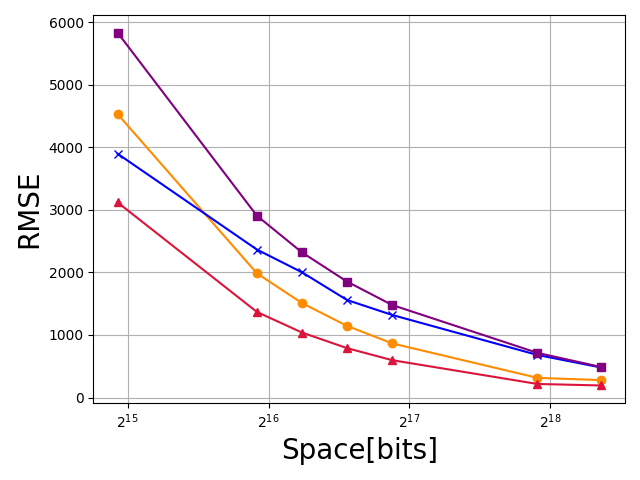}\label{fig:frequency_mem_web}} &
        \subfloat[IP]{\includegraphics[width=\smatrixCellWidth]{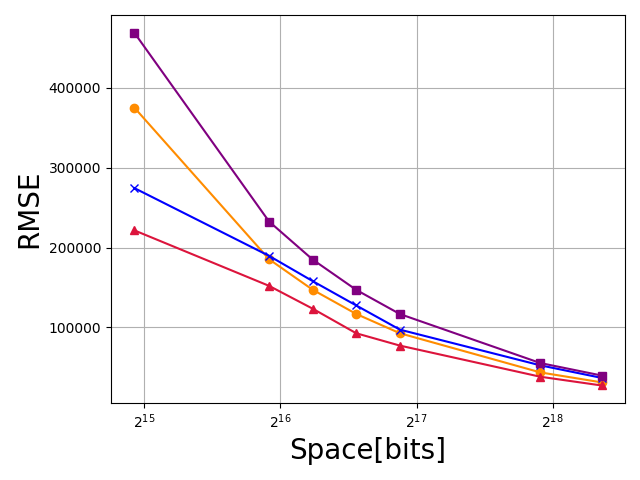}\label{fig:frequency_mem_ip}} &
        \subfloat[Zipf $\alpha=1.3$]{\includegraphics[width=\smatrixCellWidth]{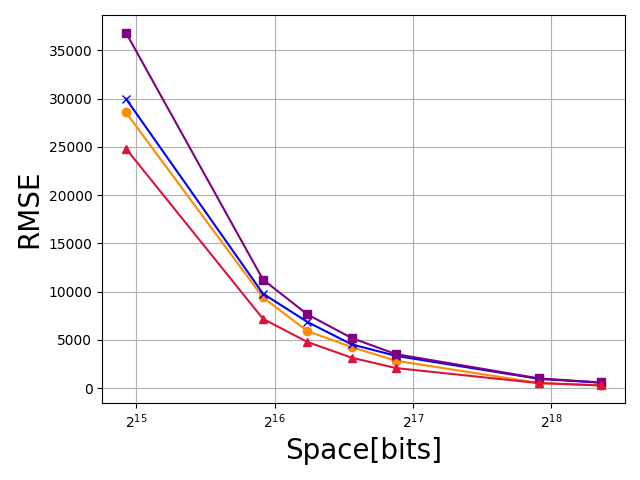}\label{fig:frequency_mem_zipf}} &
        \subfloat[Update ]{\includegraphics[width=\smatrixCellWidth]{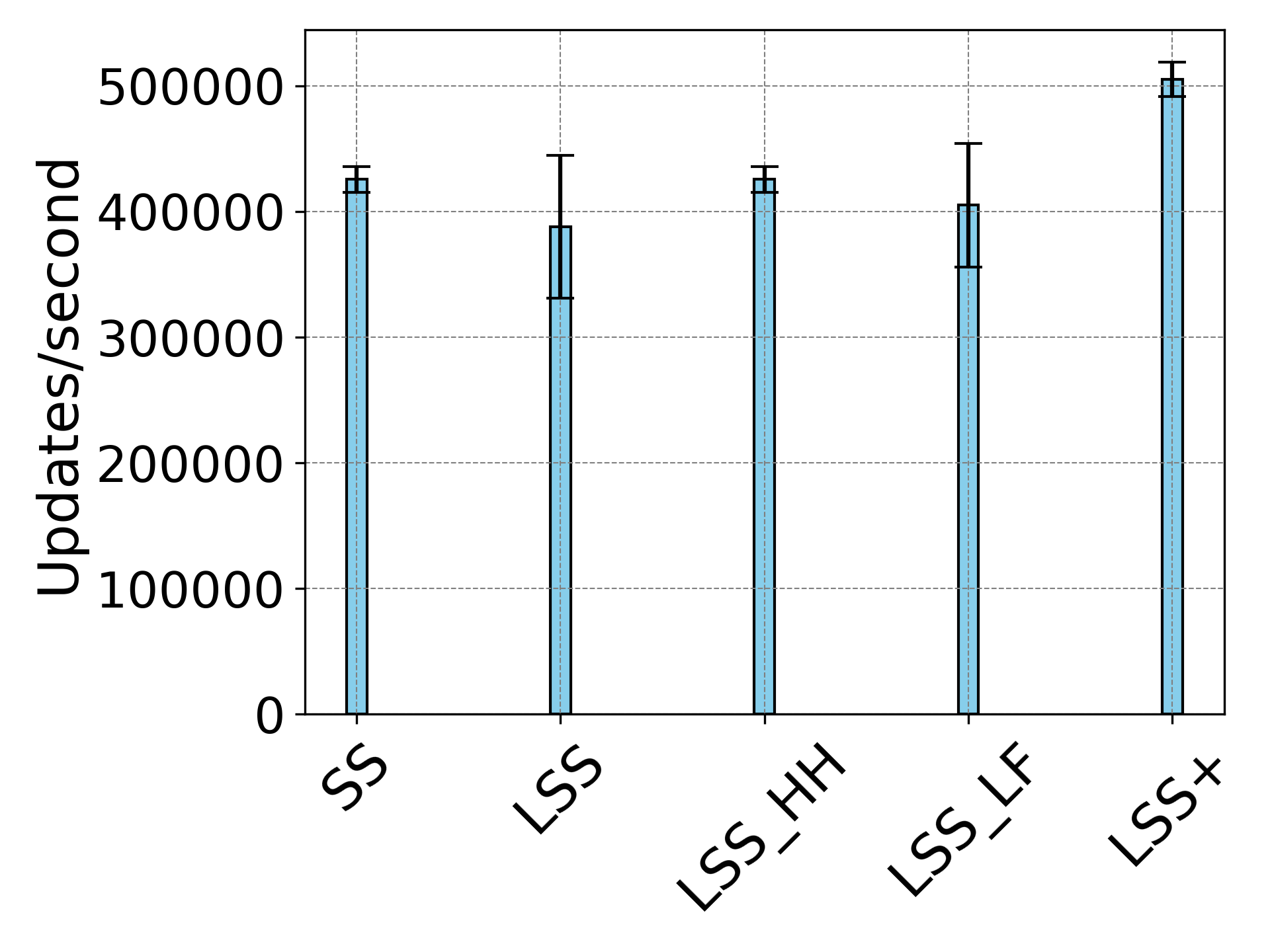}\label{fig:update_perf}} &
	\subfloat[as function of $\tau$]{\includegraphics[width=\smatrixCellWidth]{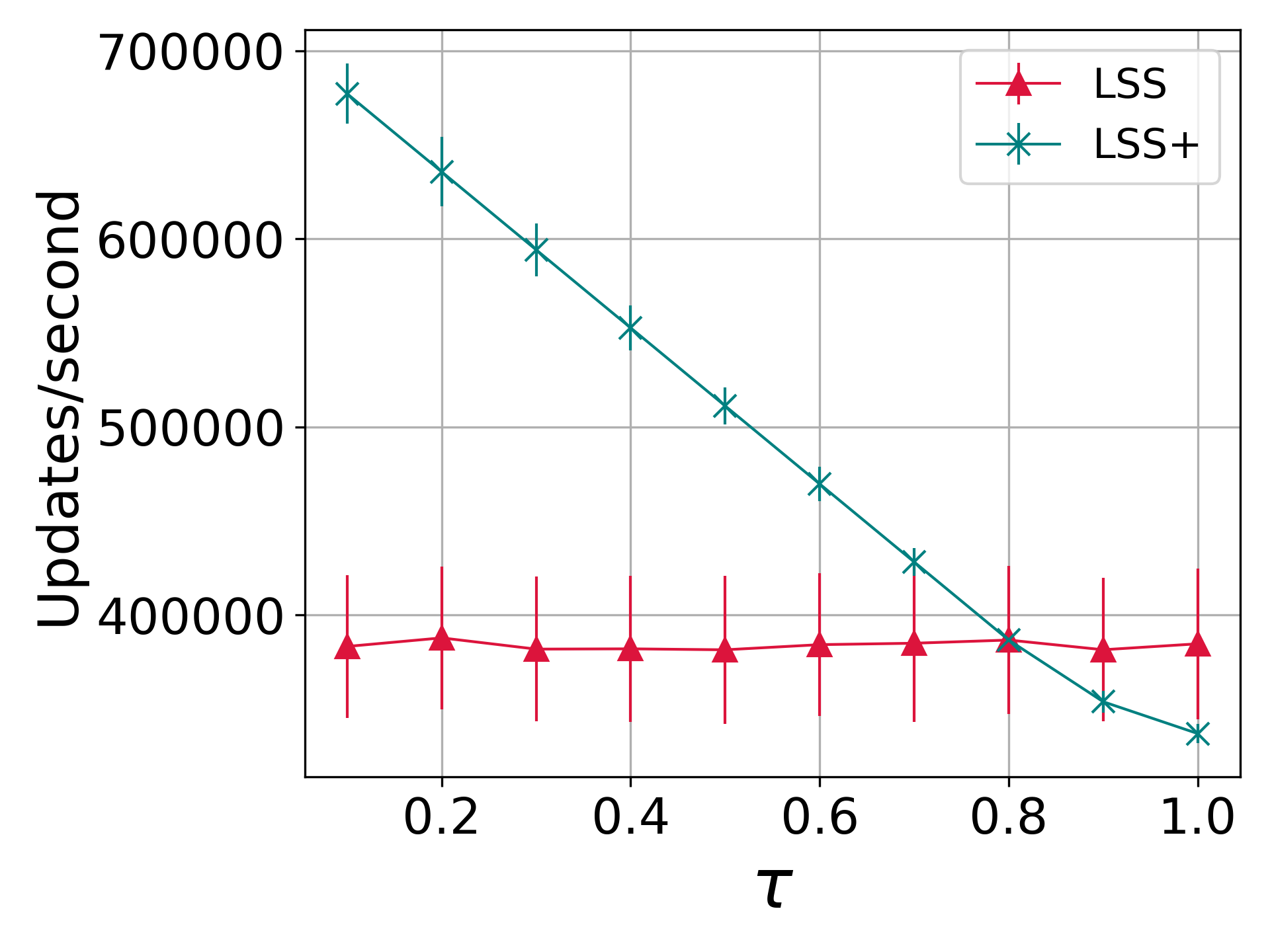}\label{fig:update_tau_perf}}
        \\
        \multicolumn{3}{c}{\subfloat{\includegraphics[width=\legendCellWidth]{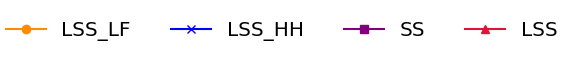}}}
    \end{tabular}
    \caption{RMSE vs. memory of frequency estimation using web search, IP and Zipf ($\alpha=1.3$) datasets.}
    \label{fig:acc_memory_frequency}
\end{figure*}

\begin{comment}
For this experiment, because assigning entries is supposed to improve accuracy specifically for the heavy hitters, we evaluated the accuracy on both queries for the ground truth heavy hitters and for all items.  
Here we use a $\theta= 2^{-5}$. Figure~\ref{fig:lss_acc_asgnentries} shows the RMSE for \lssasn{} compared to \lss{} as a function of number of assigned entries.
For this experiment, we used a perfect predictor (without adding error) and examined the change in accuracy as we varied the number of counters that handle assigned entries. We allocated $k=32$ counters to both \lss{} and \lssasn{}. In \lssasn{}, the number of assigned entries varies from 1 to $k$.
Figure~\ref{fig:lssasn_mispredictions} shows the RMSE for \lssasn{} when the heavy hitter predictor has errors w.p. $0.3$.
Assigning entries to heavy hitters improves the accuracy of queries, particularly for queries on heavy hitters. However, with an erroneous predictor and without limiting the number of assigned entries, in theory the estimation error is unbounded, and in practice we see it can grow large.
Our experiments suggest in some cases using predictions for heavy hitters and assigned entries can be valuable, but the lack of theoretical guarantees and the variance in performance from the prediction error and number of entries makes using such predictions more risky.  
\end{comment}

%{\bf MM: Is "filter ratio" clearly defined/explained somewhere?}

\paragraph{Accuracy vs. Filter Size}
Figure~\ref{fig:topk_filtersize_ip} presents the impact of the filter ratio on the precision of top-k ($k=64$) items using the IP dataset. In this experiment, we keep the threshold $t$ fixed to default ($t=4$) and vary only the filter ratio. 
%{\bf MM: do we need this next sentence? 
 %Isn't that just definitional?}
%As expected, allocating more memory space for the filter leads to a reduction in the available space for the actual counter. 
Allocating less space to the filter does not affect \lss{}'s correctness but it leads to a higher false positive rate for the filter. As a result, more items are included in the space-saving table, which affects accuracy and makes the approach more similar to falling back on the traditional space-saving algorithm.

\begin{figure*}[t]
	\centering
		\begin{tabular}{ccc}
			\subfloat[IP, Top-k]{\includegraphics[width=\mmatrixCellWidth]{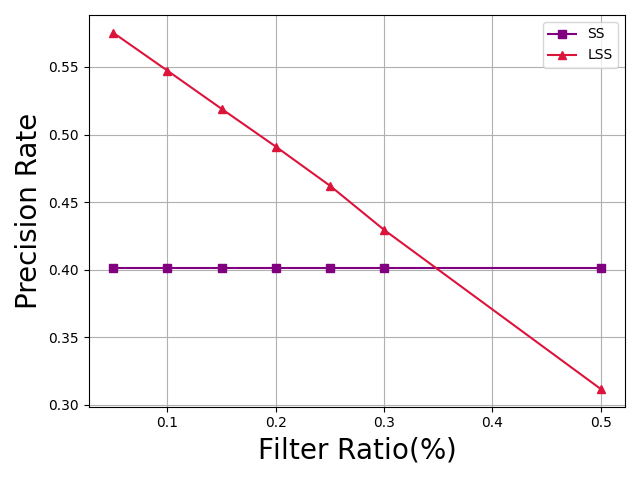}\label{fig:topk_filtersize_ip}} &
                \subfloat[Web search, HH]{\includegraphics[width=\mmatrixCellWidth]{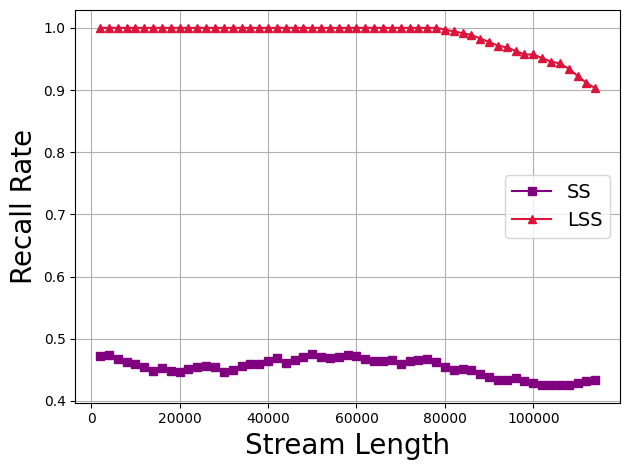}\label{fig:hh_streamlen_web}}
                \subfloat[Synthetic, Frequency]{\includegraphics[width=\mmatrixCellWidth]{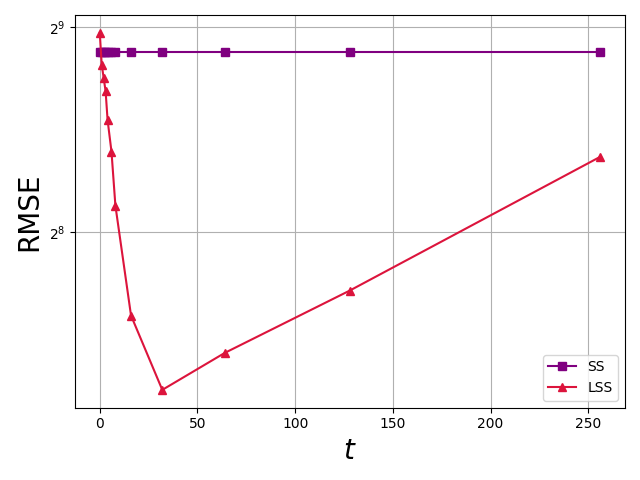}\label{fig:synthetic_vs_t}}
                \subfloat[Synthetic, Top-k]{\includegraphics[width=\mmatrixCellWidth]{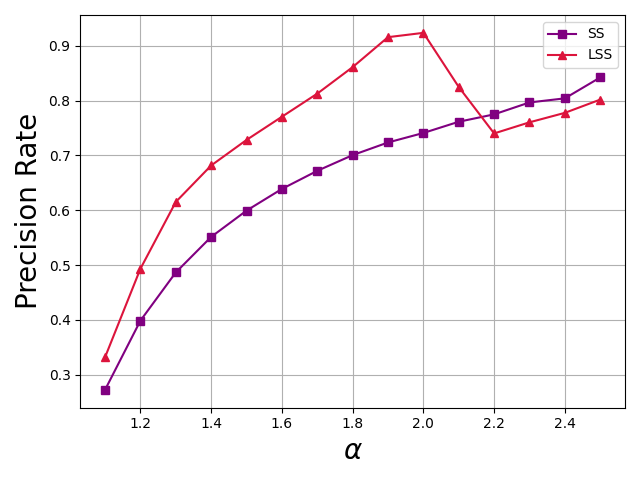}\label{fig:synthetic_vs_alpha}}
		\end{tabular}
		
	\caption{Impact of parameters (a) Precision vs. filter ratio of identifying top-k frequent item using IP dataset (b) recall rate for identifying heavy hitters in a web search dataset, focusing on the initial part of the stream (c) RMSE vs. t using synthetic dataset ($\alpha=1.3$) (d) Precision vs. $\alpha$ of identifying top-k frequent item using synthetic dataset (d) Update operation runtime using IP dataset. We use the BF implementation from~\cite{FastBF}, $t=1$ for \lss{} and \lsscbf{}, $\tau=0.5$ for \lssplus{}; (e) in relation with $\tau$.}
	\label{fig:prameters_impact}
\end{figure*}

\begin{comment}
Figure~\ref{fig:lss_acc_t} presents the RMSE of \lsscbf{} as a function of the threshold $t$, compared with \lss{} and SS. The three algorithms are compared with equal memory usage ($0.06$ Mbyte).
Similarly, intuitively
removing less significant items improves the RMSE for \lsscbf{}, but adding $t$ to the query mechanism can adversely affect accuracy, particularly at high $t$ values.
For $t=1$, \lsscbf{} exhibits marginally lower accuracy than \lss{}. The reduction in performance is caused by the higher memory consumption of the CBF over the BF, as both algorithms eliminate the same items.

\begin{figure*}[t]
	\center{
		\begin{tabular}{ccc}
			\subfloat[Web search]{\includegraphics[width=\smatrixCellWidth]{graphs/hh/filtersize/aol_hh_recall_vs_filtersize_queryOnW.png}\label{fig:hh_t_web}} &
			\subfloat[IP]{\includegraphics[width=\smatrixCellWidth]{graphs/todo.png}\label{fig:hh_t_ip}}
			\subfloat[Zipf $\alpha=1.3$]{\includegraphics[width=\smatrixCellWidth]{graphs/todo.png}\label{fig:hh_t_zipf}}

		\end{tabular}
		}
	\caption{hh, recall}
	\label{fig:hh_t}
\end{figure*}

\end{comment}

\paragraph{Accuracy vs. Stream Length} 
Figure~\ref{fig:hh_streamlen_web} shows the recall of finding heavy hitters using $2^{17}$ bit memory using web search dataset. At the beginning of the stream, LSS maintains an accurate result (high recall rate 1) by effectively filtering out low-frequency items. As the stream grows larger, medium 
items also accumulate, leading to a decrease in the recall rate for both methods. %Although LSS still outperforms SS throughout the whole stream, as shown in previous graphs, when averaging across all query intervals, the overall recall of LSS surpasses that of SS.
\paragraph{Accuracy vs. $t$} 
Figure~\ref{fig:synthetic_vs_t} displays the RMSE vs. t using a synthetic dataset with $\alpha = 1.3$ and $p=0.9$. The RMSE decreases until a certain point and then increases. This behavior is due to the increasing false positive rate of the filter at larger values of t, which is related to the number of low-frequency items and the filter size. In general, with larger memory (here we used $2^{19}$ bit memory), the filter size can be increased proportionally, allowing it to handle more low-frequency items. However, the number of low-frequency items depends on the data distribution. If prior knowledge of the distribution is available, the filter size and $t$ can be adjusted.% accordingly.

\paragraph{Accuracy vs. $\alpha$} 
Using synthetic datasets with $p=0.9$, Figure~\ref{fig:synthetic_vs_alpha} shows the precision of top-k ($k=124$) vs. $\alpha$, the skewness parameter of a Zipfian distribution. As $\alpha$ increases, the distribution becomes more skewed, with a higher concentration of low-frequency items (``heavier tail''). \lss{} has improvements over SS until a certain point ($\alpha=2$). After this point, \lss{}'s precision starts to decrease due to the saturation of the filter with low-frequency items, resulting in higher false positive rates. When $\alpha \ge 2.2$, \lss{} has lower precision than SS since the filter becomes ineffective and fewer counters are allocated to the Space Saving table compared to SS. %Again, prior knowledge of the data distribution could guide the tuning of the filter size and t.

\paragraph{Performance Comparison}
Figure~\ref{fig:update_perf} examines the update performance of SS, \lss{}, \lssasn{}, \lsscbf{} and \lssplus{} algorithms using the IP dataset. We set $t=1$ and use~\cite{FastBF} and have not optimized further.
A key consideration in comparing \lss{} to SS is the computational cost of inserting elements into the Bloom filter versus integrating them into the Space-Saving data structure.
When an item is inserted or queried within a Bloom filter, additional hash computations take place. 
The performance of \lss{} degrades slightly due to the fact that insertion into the Bloom filter is less efficient than updating the Space-Saving data structure.
Meanwhile, \lssplus{}, configured with $\tau=0.5$, noticeably outperforms both the aforementioned versions. This superior performance is due to its ability to minimize the number of insertions to the Bloom filter and to the SS. We skip query speed below since the discussed algorithms have the same query process. Figure~\ref{fig:update_tau_perf} examines the update performance of \lss{} and \lssplus{} as a function of the parameter $\tau$. (These experiments all use 32-bit floating point counters.) As $\tau$ increases, \lssplus{} saves more Bloom filter operations, resulting in improved update performance up to $\tau=0.8$.
Following this, \lssplus{} shows a slight drop in performance compared to \lss{}.

\section{Related Works}
\label{sec:related}

There are many algorithms proposed in the literature for frequency estimation. top-k, the frequent elements problem, and their variations. See~\cite{cormode2008finding} for a survey. 
Algorithms for these problems fall into two main classes~\cite{metwally2005efficient}: (deterministic) competing-counter-based techniques and (randomized) hashing-based techniques.
competing-counter-based techniques (e.g. Space Saving) maintain a separate counter for each item within the monitored set, a subset of the stream. The counters for monitored items are updated when they appear in the stream. In contrast, when there is no counter for the observed item, the item is either ignored or some algorithm-dependent action is taken.
hashing-based techniques (e.g. Count-Min Sketch~\cite{cormode2005improved}) use competing-counter-based bitmaps to estimate all items' frequencies rather than monitoring a subset. Each item is hashed into a space of counters using a family of hash functions, and every arrival within the stream updates the counters.
~\cite{li2023ladderfilter} proposed discarding approximately infrequent items in the entire data stream setup using multiple LRU queues that code the item IDs. This approach is based on the assumption that items that have been infrequent for a long period are unlikely to become frequent later. However, this assumption may not hold true for every dataset.

%\rana{sliding window sketches}
\begin{comment}
The problem of estimating item frequencies over sliding windows was first studied in~\cite{ArasuM04}.
To estimate frequency within a $W\epsilon$ additive error over a window of size $W$, their algorithm requires $O(\oneOverE\log^2\oneOverE\log W)$ bits of memory. This memory requirement was later optimized to $O(\oneOverE\log W)$ bits as highlighted in~\cite{LeeT06}. Hung and Ting in~\cite{HungLT10} further refined this by improving the update time to a constant and locating all heavy hitters in the optimal $O(\oneOverE)$ time. The WCSS algorithm, as introduced in~\cite{ben2016heavy}, also provides frequency estimates in constant time. 
\end{comment}

Algorithms with predictions is, as we have stated, a rapidly growing area.  The site \cite{awpweb} contains a collection of over a hundred papers on the topic.  
The seminal work~\cite{kraska2018case} focused on employing machine learning to refine indexing data structures and devise enhanced algorithms for data management, and there are now numerous works on this theme. For example,~\cite{mitzenmacher2018model} analyzed and enhanced learned Bloom filters as described in~\cite{kraska2018case}. The work presented in~\cite{vaidya2022snarf} introduced a learned range filter for range queries on numerical datasets. The study in~\cite{sabek2022can} explored the potential benefits of replacing traditional hash functions with learned models, aiming to minimize collisions and boost performance.

The idea of using predictions to specifically improve frequency estimation algorithms appears to have originated with \cite{hsu2019learning}, where they augmented a learning oracle of the heavy hitters into frequency estimation hashing-based algorithms. Later~\cite{jiang2019learning} explored the power of such an oracle, showing that it can be applied to a wide array of problems in data streams.

%\rana{add other learned frequency estimations}
%\rana{LadderFilter: https://dl.acm.org/doi/abs/10.1145/3588690}

%\rana{algorithms with predictions}
\section{Conclusion}
\label{sec:conclusion}

%Frequency estimation and identifying frequent items are fundamental tasks for real-world systems.

Identifying heavy hitters and estimating the frequencies of flows are fundamental tasks in various network domains. Recent works have explored the use of machine learning techniques to enhance algorithms for approximate frequency estimation problems.
However, these studies have focused only on the hashing-based approach, which may not be best for identifying heavy hitters. In this work, we have presented a novel learning-based approach for identifying heavy hitters, top $k$, and flow frequency estimation.
We have applied this approach to the well-known Space Saving algorithm, which we have called Learned Space Saving (LSS).
Our approach is designed to be resilient against prediction errors, as machine learning methods are inherently imperfect and may exhibit errors. We have demonstrated the benefits of our design both analytically and empirically. Experimental results on real-world datasets highlight that LSS achieves higher recall and precision rates, as well as improved root mean squared error (RMSE) compared to the traditional Space Saving~algorithm.

%\textbf{Code Availability}: All code is available online~\cite{opensource}

\section*{Acknowledgments}
We thank Sandeep Silwal and ChonLam Lao for their assistance with the evaluation setup.
Rana Shahout was supported in part by Schmidt Futures Initiative and Zuckerman Institute. Michael Mitzenmacher was supported in part by NSF grants CCF-2101140, CNS-2107078, and DMS-2023528.

\begin{comment}

\begin{acks}
 This work was supported by 
\end{acks}
\end{comment}

\clearpage

%\newpage
\bibliographystyle{plain}
\bibliography{refs}

\end{document}